\documentclass[sigconf]{acmart}
\AtBeginDocument{%
  }
\AtBeginDocument{\everymath{\small}}    
\usepackage{etoolbox}
\AtBeginEnvironment{align}{\small}

\copyrightyear{2026}
\acmYear{2026}
\setcopyright{cc}
\setcctype{by}
\acmConference[WWW '26]{Proceedings of the ACM Web Conference 2026}{April 13--17, 2026}{Dubai, United Arab Emirates}
\acmBooktitle{Proceedings of the ACM Web Conference 2026 (WWW '26), April 13--17, 2026, Dubai, United Arab Emirates}
\acmPrice{}
\acmDOI{10.1145/3774904.3792355}
\acmISBN{979-8-4007-2307-0/2026/04}

\usepackage{tabularx}
\usepackage{amsthm}
\usepackage{xspace}
\usepackage{mathrsfs}

\usepackage{xspace}
\usepackage{enumitem}
\usepackage{multirow}
\usepackage{bbm}
\usepackage{makecell}
\usepackage{tablefootnote}
\usepackage{threeparttable}
\usepackage{tabularx}
\usepackage{multirow}
\usepackage{graphicx}
\usepackage{subcaption}
\usepackage{bm}
\usepackage[table]{xcolor}
\usepackage{tabularx}
\usepackage{mathtools}

\usepackage{algorithm}
\usepackage[noend]{algpseudocode}
\usepackage{caption}
\usepackage{bbding}
\usepackage[dvipsnames]{xcolor}
\renewcommand{\qedsymbol}{\hfill$\blacksquare$}

\newcommand{\cmark}{\color{ForestGreen}{\CheckmarkBold}} 
\newcommand{\xmark}{\textcolor{red}{\XSolidBrush}}   

\newcommand{\smallsection}[1]{\noindent\textbf{#1.}}

\newcommand{\mat}[1]{\mathbf{#1}}
\newcommand{\matt}[1]{\mathbf{#1}^{\top}}

\newcommand{\vect}[1]{\mathbf{#1}}

\newcommand{\card}[1]{\lvert #1 \rvert}
\newcommand{\set}[1]{#1}

\newtheorem{theorem}{Theorem}
\newtheorem{lemma}[theorem]{Lemma}

\newtheorem{problem}{Problem}

\let\oldnl\nl
\newcommand{\nonl}{\renewcommand{\nl}{\let\nl\oldnl}}

\setlist[itemize]{topsep=1pt}
\setlist[enumerate]{topsep=1pt}
\newcommand{\method}{\textsc{AlphaFree}\xspace}
\newcommand{\movie}{\texttt{Movie}\xspace}
\newcommand{\book}{\texttt{Book}\xspace}
\newcommand{\video}{\texttt{Video}\xspace}
\newcommand{\baby}{\texttt{Baby}\xspace}
\newcommand{\steam}{\texttt{Steam}\xspace}
\newcommand{\beauty}{\texttt{Beauty}\xspace}
\newcommand{\health}{\texttt{Health}\xspace}

\newcommand{\items}{\bm{\mathcal{I}}}
\newcommand{\interactions}{\set{I}}
\newcommand{\collection}{\bm{\mathcal{H}}}
\newcommand{\batch}{\bm{\mathcal{B}}}
\newcommand{\traincol}{\collection_{\textnormal{train}}}
\newcommand{\lr}[1]{\vect{z}_{#1}}
\newcommand{\emb}[1]{\vect{e}_{#1}}
\newcommand{\cand}[1]{\mathcal{S}_{#1}}
\newcommand{\dlm}{d_{\textnormal{LR}}}
\newcommand{\tlm}{T_{\textnormal{LM}}}

\newcommand{\simf}[1]{\texttt{sim}_{#1}}

\begin{document}

\title{AlphaFree: Recommendation Free from Users, IDs, and GNNs}


\author{Minseo Jeon}
\orcid{0009-0007-2185-5031}
\affiliation{%
  \institution{Soongsil University}
  \city{Seoul}
  \country{Republic of Korea}}
\email{minseojeon@soongsil.ac.kr}

\author{Junwoo Jung}
\orcid{0009-0005-0612-4906}
\affiliation{%
  \institution{Soongsil University}
  \city{Seoul}
  \country{Republic of Korea}}
\email{junwoojung@soongsil.ac.kr}

\author{Daewon Gwak}
\orcid{0009-0008-5404-6553}
\affiliation{%
  \institution{Soongsil University}
  \city{Seoul}
  \country{Republic of Korea}}
\email{daewon1@soongsil.ac.kr}

\author{Jinhong Jung}
\authornote{Corresponding author.}
\orcid{0000-0002-5533-1507}
\affiliation{%
  \institution{Soongsil University}
  \city{Seoul}
  \country{Republic of Korea}}
\email{jinhong@ssu.ac.kr}

\renewcommand{\shortauthors}{Minseo Jeon, Junwoo Jung, Daewon Gwak, and Jinhong Jung}







\begin{abstract}
    Can we design effective recommender systems free from users, IDs, and GNNs?
Recommender systems are central to personalized content delivery across domains, with top-$K$ item recommendation being a fundamental task to retrieve the most relevant items from historical interactions.
Existing methods rely on entrenched design conventions, often adopted without reconsideration, such as storing per-user embeddings (user-dependent), initializing features from raw IDs (ID-dependent), and employing graph neural networks (GNN-dependent).
These dependencies incur several limitations, including high memory costs, cold-start and over-smoothing issues, and poor generalization to unseen interactions.

In this work, we propose \method, a novel recommendation method free from users, IDs, and GNNs.
Our main ideas are to 
infer preferences on-the-fly without user embeddings (user-free), 
replace raw IDs with language representations (LRs) from pre-trained language models (ID-free), and 
capture collaborative signals through augmentation with similar items and contrastive learning, without GNNs (GNN-free).
Extensive experiments on various real-world datasets show that \method consistently outperforms its competitors, achieving up to around 40\% improvements over non-LR-based methods and up to 5.7\% improvements over LR-based methods, while significantly reducing GPU memory usage by up to 69\% under high-dimensional LRs.

\end{abstract}




\maketitle

\section{Introduction}
\label{sec:intro}
\label{sec:introduction}
\begin{figure}
    \centering
    \includegraphics[width=0.95\linewidth]{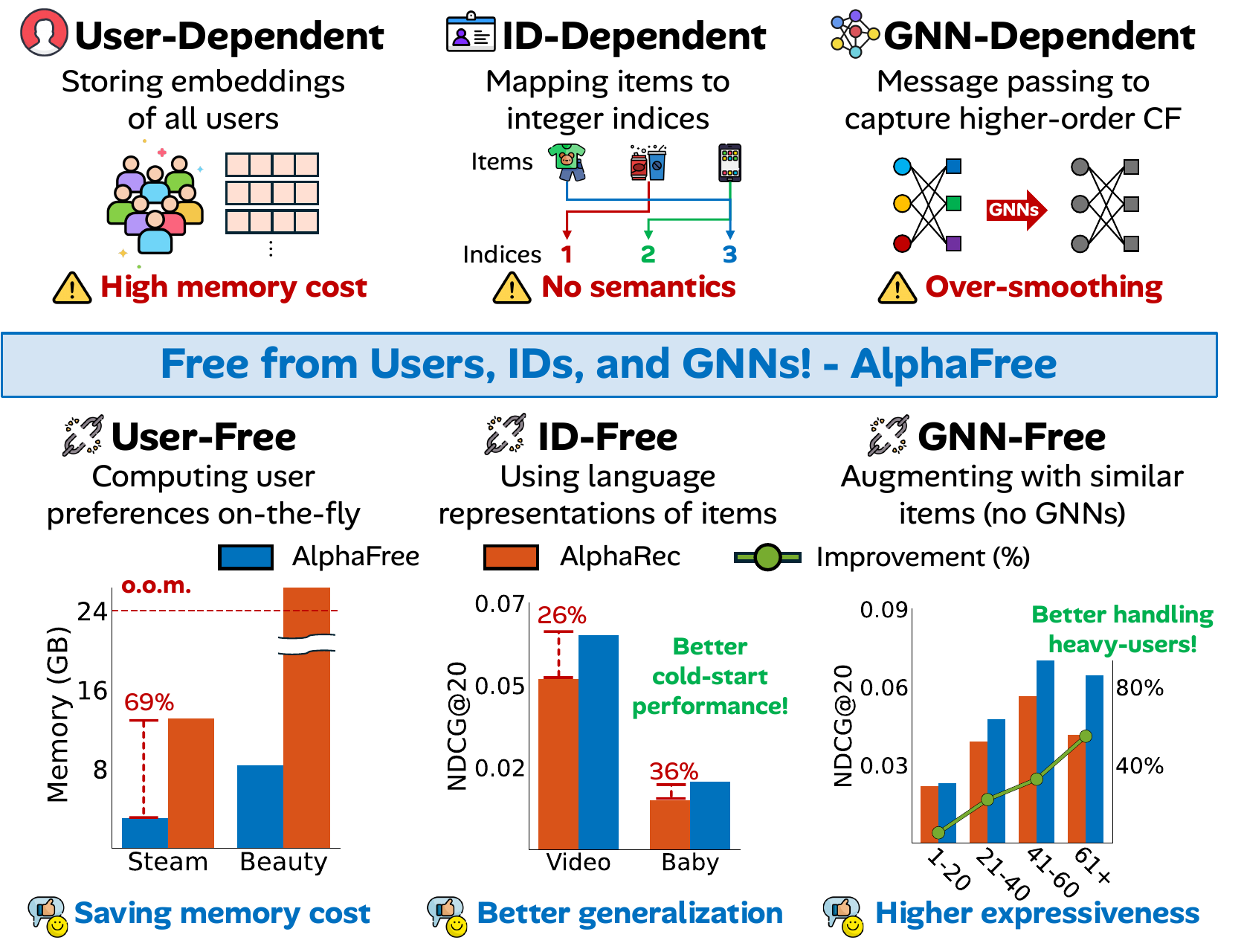}
    \caption{
        Highlight of the motivation and effectiveness of \method.
        (Top) Existing recommender systems depend on user embeddings, item IDs, or GNNs, leading to high memory cost, cold-start limitations, and over-smoothing.
        (Bottom) \method overcomes these challenges by performing recommendations free from users, IDs, and GNNs, showing 69\% less memory, 26–36\% higher accuracy in cold-start cases, and up to around 50\% improvement for heavy users.
    }
    \label{fig:introduction}
\end{figure}

Recommender systems play a central role in online services by delivering personalized content across domains such as e-commerce, streaming, and social media~\cite{abs-2407-13699,ParkJK17,LeeK25}.
A fundamental task in this field is \textit{top-$K$ item recommendation}, which aims to retrieve the $K$ most relevant items based on a user’s historical interactions~\cite{LeeSKLL11,JungPSK17,JungJSK16}.
Over the years, numerous methods have been proposed for this task, ranging from classical collaborative filtering~\cite{RendleFGS09, Rendle10, JuanZCL16, KorenBV09, Steck19} and neural network-based approaches~\cite{HeLZNHC17,0001DWLZ020, YuY00CN22, Wang0WFC19,LeeKS24} to more recent language model-based techniques~\cite{Geng0FGZ22, Sheng0ZCWC25, RenWXSCWY024, BaoZZWF023}, all striving to improve performance.
Throughout this development, recommender modeling has been largely shaped by several conventions that have become de facto standards: storing user embeddings, initializing with ID-based features, and employing graph neural networks (GNNs) to capture collaborative signals.
Based on these aspects, we categorize existing approaches into (1) user-dependent methods, (2) ID-dependent methods, and (3) GNN-dependent methods.

\textit{User-dependent methods} extract user and item embeddings from historical interactions, and require user embeddings to be maintained throughout training and inference.
Although user embeddings encode individual preferences, they become unreliable when interactions are sparse~\cite{abs-2312-10073} or noisy~\cite{Gao0HCZFZ22, HuangDDYFW021}, and their memory cost scales linearly with the number of users, which becomes substantial in domains with many users, such as e-commerce.
\textit{ID-dependent methods} rely on raw discrete identifiers (IDs) as the initial features of users and items. While simple and effective for seen entities, these representations lack semantic meaning and generalize poorly to unseen users or items.
\textit{GNN-dependent methods} use GNNs to capture high-order collaborative signals from explicit graph structures (e.g., user–item graphs).
However, they also need to preserve user embeddings (thus inherently user-dependent) as well as graphs, which leads to high memory usage, and their expressiveness is limited due to over-smoothing~\cite{LiHW18, abs-2303-10993, ChenLLLZS20}.

Several prior works have attempted to partially address the limitations of these three methods.
For instance, FISM~\cite{KabburNK13} and SLIM-based methods~\cite{NingK11} neither store user embeddings nor employ GNNs; however, they rely on ID-based features and struggle to capture higher-order collaborative signals.
LightGCN~\cite{0001DWLZ020} and XSimGCL~\cite{YuXCCHY24} capture higher-order structures with GNNs but also rely on ID-based features, store user embeddings, and suffer from over-smoothing.
Recent methods such as RLMRec~\cite{RenWXSCWY024} and AlphaRec~\cite{Sheng0ZCWC25} leverage language models to exploit semantic representations of items, but they remain user-dependent or rely on GNNs.
Figure~\ref{fig:introduction} summarizes these dependencies and limitations, motivating our work, while Table~\ref{table:overall} shows that no existing approach addresses them simultaneously, raising the question: \textit{Can we design effective recommender systems free from users, IDs, and GNNs?}

In this paper, we propose \method, a novel method to top-$K$ item recommendation that is free from users, IDs, and GNNs, while achieving strong performance across various benchmarks.
\method builds on three core design principles:
\begin{itemize}[leftmargin=4mm]
    \item{\textbf{User-free}: Rather than storing user embeddings, it computes latent preferences on-the-fly from interaction history, reducing memory usage and improving cold-start handling.}
    \item{\textbf{ID-free}: It replaces raw IDs with semantic item representations derived from language models, enabling generalization to unseen items and capturing semantic similarity.}
    \item{\textbf{GNN-free}: It captures collaborative signals without GNNs by augmenting representations with behaviorally and semantically similar items and applying contrastive learning.
    }
\end{itemize}

By removing all three dependencies, \method achieves greater memory efficiency, improved generalization, and higher accuracy than the existing methods.
Our contributions are as follows:
\begin{itemize}[leftmargin=4mm]
    \item {\textbf{Method.} We propose \method, a novel method that performs recommendations free from users, IDs, and GNNs.}
    \item {\textbf{Techniques.} We show our design principles are effectively realized through:
    (1) augmenting with behaviorally and semantically similar items,
    (2) on-the-fly preference modeling, and
    (3) contrastive learning between the original and augmented views.
    }
    \item {\textbf{Experiment.} 
    Through our extensive experiments, we show \method achieves better recommendation accuracy and memory efficiency compared to state-of-the-art methods.
    } 
\end{itemize}

\vspace{-2.5mm}
\section{Related Work}
\label{sec:related}
\smallsection{Traditional methods dependent on users, IDs, or GNNs} 
Most classical methods are \textit{user-dependent}, as they require user embeddings.
Matrix factorization-based methods, such as MF-BPR \cite{KorenBV09,RendleFGS09} and NCF~\cite{HeLZNHC17}, extract user embeddings along with item embeddings from a user-item  matrix.
Graph-based methods, such as LightGCN~\cite{0001DWLZ020} and XSimGCL~\cite{YuXCCHY24}, perform message passing over user-item graphs to learn users and items embeddings.
Although they directly encode user preferences in an embedding space, they incur memory costs that grow with the number of users, 
which becomes substantial in user-heavy domains such as e-commerce (see Table~\ref{tab:datasets}).

Moreover, traditional methods including the aforementioned models  are \textit{ID-dependent}, initializing embeddings with raw discrete IDs (e.g., unique integer indices for users or items) due to the absence of user or item content and their reliance on user-item interactions.
These IDs are typically mapped to one-hot encoding or randomly initialized features followed by training.
While straightforward and effective for learning interactions without domain knowledge, they struggle to capture meaningful patterns under sparse interactions and to generalize to new entities.

With the emergence of GNNs, many methods have been developed that are \textit{GNN-dependent}, relying on message passing over a user–item graph to capture higher-order collaborative signals.
For example, LightGCN employs simplified GNN layers to propagate user and item embeddings across the graph, capturing higher-order  signals through multiple layers.
XSimGCL~\cite{YuXCCHY24} further introduces graph contrastive learning~\cite{YouCSCWS20,abs-1809-10341,abs-2006-04131} for better expressiveness.
These methods enhance performance by using higher-order CF signals, but they remain inherently user-dependent, requiring embeddings for all users. 
Additionally, they suffer from over-smoothing, where embeddings become too similar after repeated message passing, which is particularly severe for high-degree nodes due to excessive averaging over many neighbors, limiting performance.

\smallsection{Methods partially free from users, IDs, or GNNs} 
Several studies have attempted to partially relax those dependencies. 
User-free methods discard user embeddings while retaining only item embeddings. Along this line, SLIM~\cite{NingK11} directly learns item–item similarities, while FISM~\cite{KabburNK13} models them through factorized matrices of items. However, both approaches remain ID-dependent. 
Early-stage methods such as MF-BPR and NCF are trivially GNN-free as they predate the advent of GNNs, while still relying on users and IDs.
With advances in language models, RLMRec~\cite{RenWXSCWY024} and AlphaRec~\cite{Sheng0ZCWC25} have been recently proposed as ID-free, using language representations (LRs) derived from item content instead of IDs during training. 
However, both store user embeddings and thus remain user-dependent, while AlphaRec is additionally GNN-dependent.
In particular, the LR-based methods face challenges of user dependency w.r.t. memory consumption, as most language models produce high-dimensional embeddings, which are then used to compute embeddings for every user\footnote{For example, at the initial stage, LightGCN uses 64-dimensional embeddings~\cite{0001DWLZ020}, while AlphaRec adopts 3,072-dimensional ones from text-embedding-3-large~\cite{Sheng0ZCWC25,abs-2201-10005}.}.
To the best of our knowledge, no existing work resolves all these dependencies simultaneously, whereas our \method successfully resolves them.

\section{Preliminaries}
\label{sec:prelim}
\begin{figure*}[t]  
\centering
\includegraphics[width=1\textwidth]{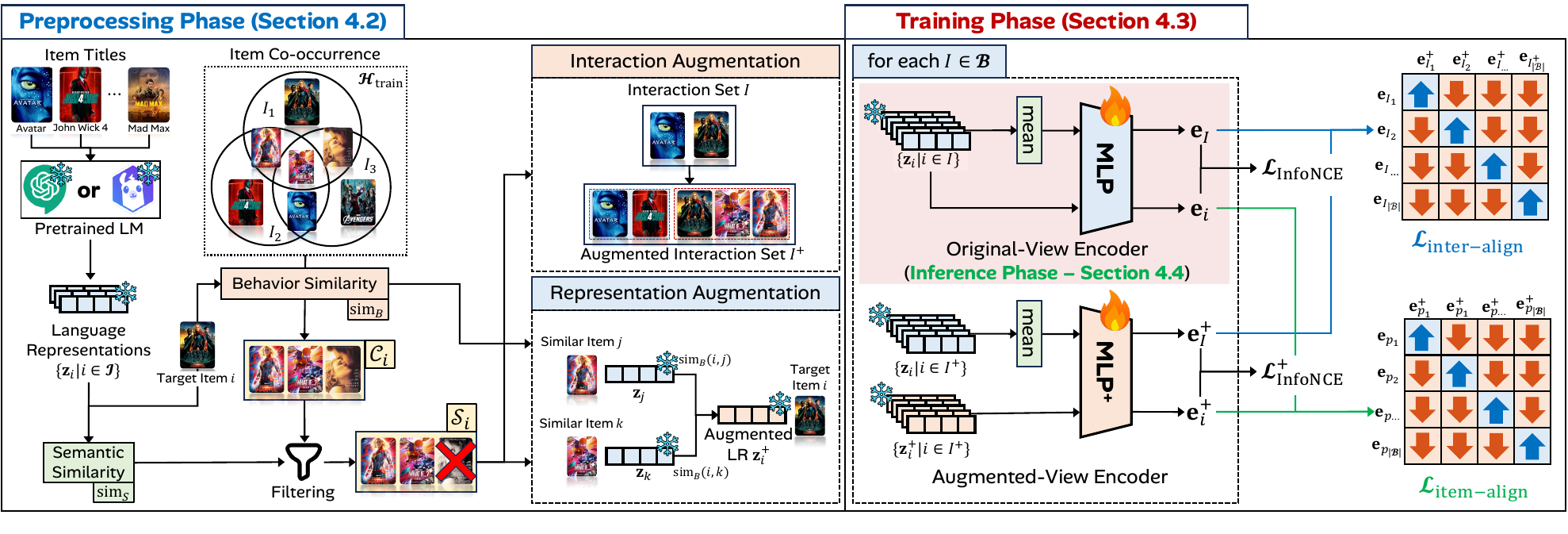}
\vspace{-7mm}
\caption{
Overall process of \method, consisting of preprocessing (item LRs and augmentations), training (contrastive alignment of the original and augmented views), and inference (using only the original-view encoder). See Table~\ref{tab:symbols} for symbols.
}
\label{fig:overall}
\end{figure*}

\smallsection{Notations}
We denote the set of all items by $\items$, and $t_i$ is the associated text content (e.g., title) of item $i$.
Let $\interactions \subseteq \items$ denote the interaction set of an arbitrary user, i.e., $\interactions = \{i_1, \cdots, i_{\card{\interactions}}\}$.
Let $\collection$ denote the collection of all interaction sets in a dataset, i.e., $\collection = \{\interactions \mid \interactions \subseteq \items\}$.

\smallsection{Problem definition} The problem is defined as follows:
\begin{problem}[Top-$K$ Item Recommendation]
\label{prob}
\textbf{Given} a training collection of interaction sets $\traincol$, 
the goal is to \textbf{learn} the parameters $\theta$ of a recommendation model,
in order to \textbf{recommend} the items in the top-$K$ ranked list $\mathcal{R}_{\interactions_q}$ for a querying interaction set $\interactions_q$:
\begin{equation*}
    \mathcal{R}_{\interactions_q} := \texttt{argtop}_{K}\Bigl(
        \bigl\{ 
            y_\theta(i \mid \interactions_q) \mid i \in \items  
        \bigr\}
    \Bigr)
\end{equation*}
where $y_{\theta}(i \mid \interactions_q)$ indicates the recommendation score of item $i$ w.r.t. $\interactions_q$ produced by the model of $\theta$, and $\texttt{argtop}_{K}(\cdot)$ returns the list of the top-$K$ items with the highest scores.
\end{problem}

\vspace{-2.5mm}
\section{Proposed Method}
\label{sec:method}
We propose \method, a novel recommendation method free from users, IDs, and GNNs for Problem~\ref{prob}.

\subsection{Overview}
\label{sec:method:overview}
We describe the overview of \method, where Figure~\ref{fig:overall} depicts its overall process.
As introduced in Section~\ref{sec:intro}, 
we design \method based on the following principles:
\begin{itemize}[leftmargin=4mm]
    \item{
        \textbf{User-free:} 
        We adopt a user-free mechanism that infers latent preferences on-the-fly from an interaction set, thereby eliminating user embeddings and improving memory efficiency\footnote{Here, “user-freeness” does not mean the removal of personalization, but rather aims to eliminate the need to store user embeddings by computing them on the fly.}.
    }
    \item{
        \textbf{ID-free:} We replace raw item IDs with language representations (LRs) derived from item contents using a pre-trained language model, which enhances generalization.
    }
    \item{
        \textbf{GNN-free:} We avoid GNNs and instead augment embeddings with similar items, improving efficiency and expressiveness by avoiding message-passing and over-smoothing.
    }
\end{itemize}

By removing these dependencies, we show that a simple and lightweight model can achieve efficient memory usage (particularly important when using high-dimensional LRs) while delivering higher recommendation accuracy across diverse benchmarks.

\method consists of three phases: preprocessing, training, and inference (refer to Appendix~\ref{appendix:algorithms} for details).
The preprocessing phase (Section~\ref{sec:method:preprocessing}) computes LRs for all items and retrieves behaviorally and semantically similar items to augment interactions and representations.
The training phase (Section~\ref{sec:method:training}) learns embeddings for items and interaction sets in batches, and injects higher-order CF signals through contrastive learning with the augmented ones.
Unlike training, the inference phase (Section~\ref{sec:method:inference}) performs recommendations solely based on the querying interaction set $\interactions_q$, without augmentations.
As the trained parameters embed information from augmented interactions, this allows it to effectively capture latent preferences in $\interactions_q$.

\subsection{Preprocessing Phase (Algorithm~\ref{alg:preprocessing})}
\label{sec:method:preprocessing}

This phase aims to preprocess LRs for all items in $\items$ and identify similar items for a target item using the LRs and $\traincol$, which are used to augment interactions and representations as higher-order CF signals.
Notice this phase is performed only once before training.

\subsubsection{Language Representations of Items}
\label{sec:method:preprocessing:lm}
Unlike ID-dependent methods, LR-based approaches such as AlphaRec~\cite{Sheng0ZCWC25} provide empirical evidence that item LRs inherently capture CF signals, implying that better performance can be achieved without relying on raw IDs.
Inspired by this, we pursue an \textbf{ID-free} design by using a pre-trained language model $\texttt{LM}(\cdot)$ as follows: 
\begin{equation}
    \label{eq:lm}
    \lr{i} \leftarrow \texttt{LM}(t_i),
\end{equation}
where $\lr{i} \in \mathbb{R}^{\dlm}$ denotes the LR of item $i$ with output dimension $\dlm$, and $t_i$ is the text description of item $i$ (i.e., its title).
Any pre-trained LM can be used for this (e.g., BERT, LLaMA, or GPT.), and the detailed configurations of the LMs are provided in Appendix~\ref{appendix:lr_explain}.



\subsubsection{Behaviorally and Semantically Similar Items}
\label{sec:method:preprocessing:similar}

We identify similar items $\cand{i}$ for each item $i$, using $\traincol$ and $\{\lr{j}\}$.
Each $I \in \traincol$ provides behavioral signals indicating which items are consumed by a user, while $\lr{j}$ captures the semantics of item $j$.
We describe our simple but effective approach using them for obtaining $\cand{i}$.

\smallsection{Behavior similarity} 
For $\traincol$, we employ \textit{co-occurrence}~\cite{Karypis01} as behavioral similarity, since it offers a simple and efficient way to capture CF signals and has been widely adopted in prior work~\cite{VeitKBMBB15,LiuZZ023,ZhouZY23}.
The intuition is that items frequently consumed together are likely to share collaborative patterns, with co-occurrence leveraging interactions as implicit bridges to capture higher-order signals.
Formally, the measure is defined as:
\begin{equation}
\label{eq:sim:b}
\simf{B}(i, j) = \!\!\!\!\! \sum_{\set{I} \in \traincol} \!\!\! \mathbb{I}(i \in \set{I}) \cdot \mathbb{I}(j \in \set{I}),
\end{equation}
where $\mathbb{I}(\cdot)$ returns $1$ if the given predicate holds, and $0$ otherwise.

\smallsection{Semantic similarity}
We also consider semantic similarity, which captures the inherent meaning of items. 
Since LRs are commonly compared via dot products~\cite{KarpukhinOMLWEC20, 0046ZL21}, we adopt this measure, as:
\begin{equation}
\label{eq:sim:s}
\simf{S}(i,j) = \lr{i}^\top\lr{j},
\end{equation}
where $\lr{i}$ and $\lr{j}$ are the LRs of items $i$ and $j$, respectively.

\smallsection{Behavioral and semantic filtering}
While $\simf{B}$ provides useful co-occurrence signals, it may include noisy candidates.
To refine the selection, we introduce our filtering strategy based on both similarities.
Specifically, we first construct a candidate set $\mathcal{C}_{i}$ by selecting the top-$K_c$ items with the highest $\simf{B}$ to item $i$, as:
\begin{equation}
\label{eq:cand:c}
\mathcal{C}_{i} = \texttt{argtop}_{K_c}\Bigl(\bigl\{
        \simf{B}(i, j) \mid j \in \items
    \bigr\}\Bigr).
\end{equation}
We then filter $\mathcal{C}_{i}$ by retaining only those whose semantic similarity exceeds the mean value $\mu_i$ as follows:
\begin{equation}
\label{eq:sim:bs}
\cand{i} = \bigl\{
    j \in \mathcal{C}_{i} \mid \simf{S}(i, j) \geq \mu_i
\bigr\}, \quad 
\mu_i = \frac{1}{\card{\items}} \sum_{j \in \items} \simf{S}(i,j).
\end{equation}
Figure~\ref{fig:symmentic_filtering} depicts the intuition behind our strategy and the choice of $\mu_i$. 
Given a query item $i$ (blue), candidates are first retrieved by $\simf{B}$ and then refined by $\simf{S}$ using $\mu_i$. 
The semantic similarity distribution shows the 25\%, $\mu_i$, and 75\% values; the 25\% cutoff is too tight, the 75\% cutoff too loose, and $\mu_i$ provides a balanced threshold that retains semantically relevant items while filtering out irrelevant ones.
The performance under these cutoff values are provided in Appendix~\ref{appendix:semantic}, where $\mu_i$ provides better performance.

\begin{figure}[t!]
    \centering
    \begin{subfigure}[b]{0.48\linewidth}
        \includegraphics[width=\linewidth]{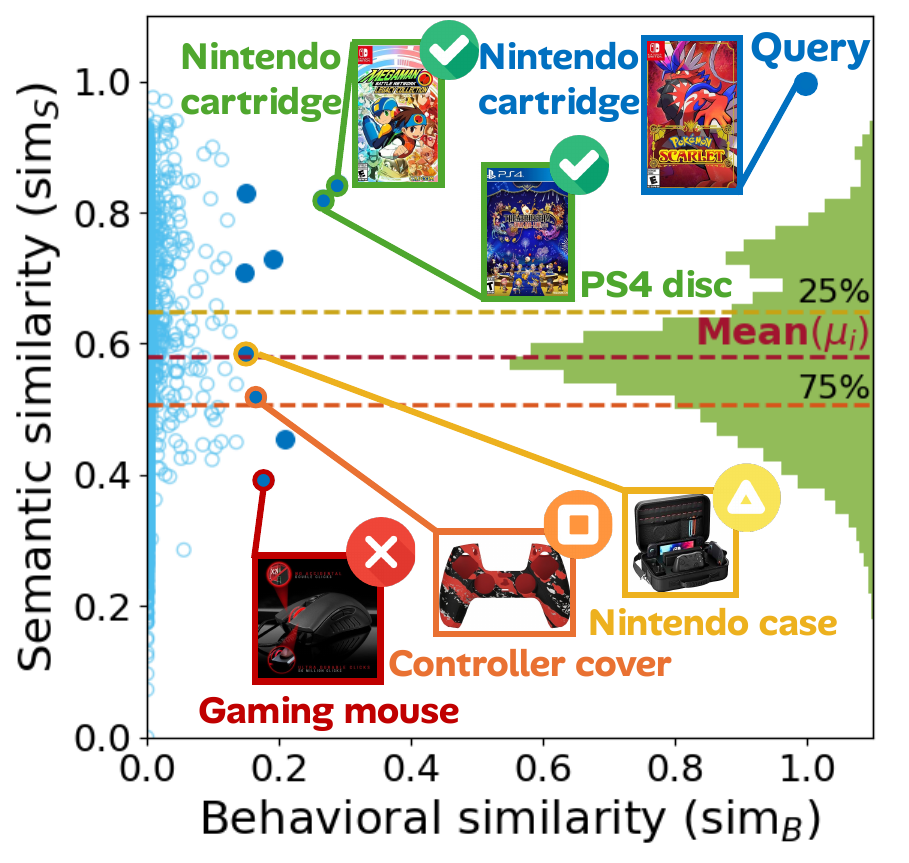}
        \caption{\video.}
        \label{fig:symmentic_filtering:video}
    \end{subfigure}
    \hfill
    \begin{subfigure}[b]{0.48\linewidth}
        \includegraphics[width=\linewidth]{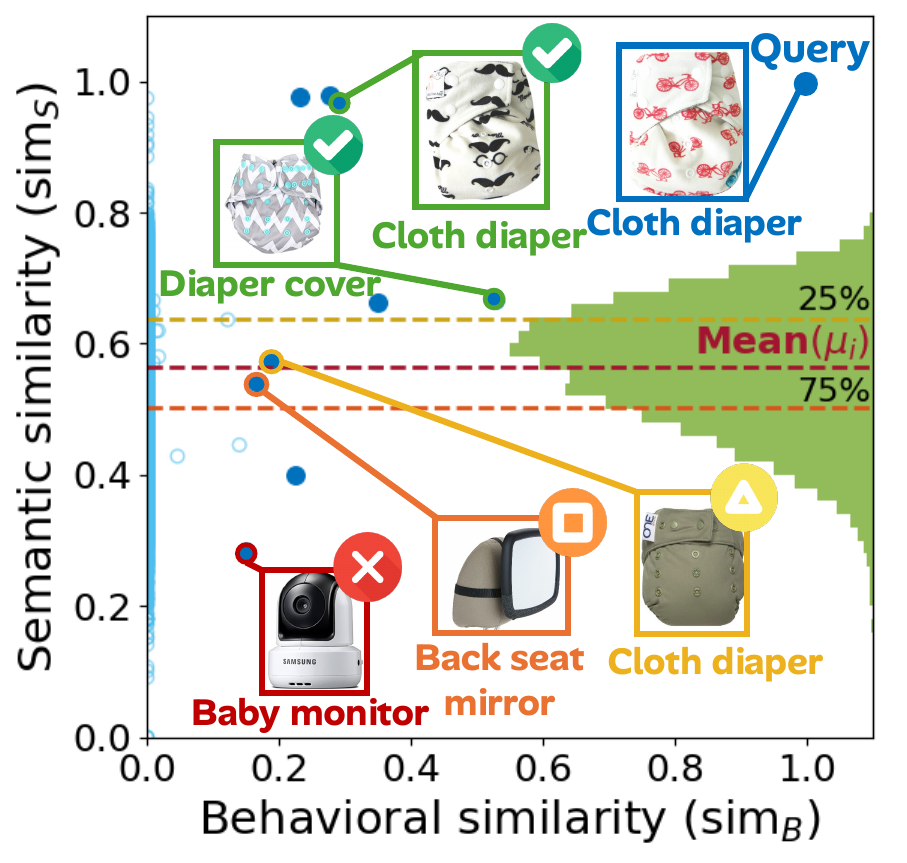}
        \caption{\baby.}
        \label{fig:symmentic_filtering:baby}
    \end{subfigure}
    \caption{
        Examples of behavioral–semantic similarities (min-max normalized) for query items $i$, where $K_c=10$, blue circles denote $\mathcal{C}_i$ based on $\simf{B}$, and the green histogram shows the distribution of $\simf{S}$.
        A 25\% threshold based on $\simf{S}$ yields overly semantically similar items, the 75\% threshold includes weakly related ones, while $\mu_i$ provides a balanced trade-off.
    }
    \label{fig:symmentic_filtering}
\end{figure}

\subsubsection{Augmentation of Interactions and Representations}
\label{sec:method:preprocessing:augmentation}

For each item $i$, we obtain its similar items $\cand{i}$, which serve as higher-order CF signals for the following augmentations, allowing our model to remain \textbf{GNN-free} in both training and inference.  

\smallsection{Interaction augmentation}
\label{sec:method:interaction_augmenatation}
To enrich CF signals, we augment each $I \in \traincol$ with similar items. 
For each item $i \in I$, we use $\cand{i}$ and take the union with $I$, yielding the augmented set:  
\begin{equation}
I^+ = I \cup \texttt{union}\Bigl(
    \bigl\{
        \cand{i} \mid i \in I
    \bigr\}
\Bigr).
\end{equation}
This injects higher-order CF signals at the \emph{interaction level}, adding new interactions that were not explicitly observed.

\smallsection{Representation augmentation}
\label{sec:method:representation_augmenatation}
We also augment the representation of each item $i$ with $\cand{i}$. 
In traditional item–item CF~\cite{SarwarKKR01,LindenSY03}, items that frequently co-occur in interactions are regarded as related and provide useful CF signals. 
Inspired by this, we combine the LRs of similar items $\cand{i}$, weighted by the behavioral similarities:  
\begin{equation}
    \label{eq:representation_aug}
    \lr{i}^{+} = \frac{\sum_{j \in \cand{i}}\simf{B}(i,j)\cdot\lr{j}}{\sum_{j \in \cand{i}}\simf{B}(i,j)}.
\end{equation}
This injects higher-order CF signals at the \emph{representation level}, enhancing $\lr{i}$ with co-occurrence-based behavioral signals.

\subsection{Training Phase (Algorithm~\ref{alg:training})}
\label{sec:method:training}
This phase takes the original and augmented views (or data), and learns model parameters $\Theta$ using mini-batches $\batch \subseteq \traincol$ at each epoch. For each $\batch$, we encode both views and contrastively align them, thereby injecting higher-order CF signals into $\Theta$.

\subsubsection{Encoders for the Original and Augmented Views} 
\label{sec:method:training:encoders}
We describe how both encoders generate embeddings for each $I \in \batch$.

\smallsection{Original view} We encode the original view by computing embeddings for both $I$ and its constituent items using an MLP layer:
\begin{equation*}
    \label{eq:encoder:original}
    \emb{I} = \texttt{MLP}\Bigl(
        \texttt{mean}\bigl(
            \{\lr{i} \mid i \in I\}
        \bigr)
    \Bigr), \;\; \text{and} \;\;
    \emb{i} = \texttt{MLP}(\lr{i}), \; \forall i \in I, 
\end{equation*}
where $\emb{I}$ captures the latent preference of an arbitrary user from interactions $I$, which is computed on-the-fly, making our model \textbf{user-free} without storing user embeddings.  

This structure can be seen as a counterpart to AlphaRec~\cite{Sheng0ZCWC25} without GNNs, but it alone cannot capture higher-order CF signals.
Thus, AlphaRec relies on GNNs to learn such signals, making it both user- and GNN-dependent.
However, \method encodes the following augmented view to avoid GNNs, resolving such dependencies.

\smallsection{Augmented view}
We encode the augmented view by computing embeddings for $I^{+}$ with augmented LRs $\{\lr{i}^{+}\}$ obtained from the preprocessing phase using another MLP layer:
\begin{equation*}
    \label{eq:encoder:augment}
    \emb{I}^{+} = \texttt{MLP}^{+}\Bigl(
        \texttt{mean}\bigl(
            \{\lr{i} \mid i \in I^{+}\}
        \bigr)
    \Bigr), \;\; \text{and} \;\;
    \emb{i}^{+} = \texttt{MLP}^{+}(\lr{i}^{+}), \; \forall i \in I^{+},
\end{equation*}
where $\emb{I}^{+}$ denotes interaction-level augmentation for $I$ and thus uses the original LRs, whereas $\emb{i}^{+}$ corresponds to representation-level augmentation for $i$ based on the augmented LRs.

\smallsection{InfoNCE losses for recommendation}
We next utilize embeddings from each view to learn recommendations. A common approach is BPR loss~\cite{RendleFGS09}, which contrasts one positive with one negative. However, recent studies~\cite{ChenTZ0SZWC24, Sheng0ZCWC25}, have empirically shown that InfoNCE, contrasting one positive with multiple negatives, outperforms BPR in recommendation. 
We thus adopt it as well.
For conciseness, we first provide a general definition of \textit{log-softmax} (or its negative is considered a local loss):
\begin{equation*}
    \ell^{(\tau)}(\vect{a}, \vect{p}; C) = \log{\frac{\texttt{exp}\bigl(\sigma(\vect{a}, \vect{p})/\tau\bigr)}{\sum_{\vect{c} \in C}\texttt{exp}\bigl(\sigma(\vect{a}, \vect{c})/\tau\bigr)}}, 
\end{equation*}
where 
$\sigma(\cdot, \cdot)$ is cosine similarity, $\tau$ is a temperature parameter, 
$\vect{a}$ and $\vect{p}$ denote embeddings for \textit{anchor} and \textit{positive} samples, respectively, 
and $C$ is the set of candidate embeddings considered in the softmax.
This quantifies the relative closeness of the anchor–positive pair against the other candidates in $C$.

For each $I \in \batch$, we randomly sample a positive item $p \sim U(I)$ and define $N_I$ as the set of negatives drawn from $\items \setminus I$.
Then, the InfoNCE loss for the original view is as follows:
\begin{equation}
    \label{eq:loss:rec:original}
    \mathcal{L}_\text{InfoNCE}(\batch) = -\frac{1}{\card{\batch}}\sum_{I\in\batch}
    \ell^{(\tau_r)}(\emb{I}, \emb{p}; \mathcal{E}_I),
\end{equation}
where $\mathcal{E}_I = \bigl\{ \emb{i} \mid i \in N_I \cup \{p\}\bigr\}$, and $\emb{i}=\texttt{MLP}(\textbf{z}_i)$ for each negative item $i \in N_\interactions$.
This encourages $\emb{I}$ to be closer to $\emb{p}$
while distinguishing it from the negatives. 
Similarly, the loss for the augmented view is represented as:
\begin{equation}
    \label{eq:loss:rec:augmented}
    \mathcal{L}_\text{InfoNCE}^{+}(\batch) = -\frac{1}{\card{\batch}}\sum_{I\in\batch}
    \ell^{(\tau_r)}(\emb{I}^{+}, \emb{p}^{+}; \mathcal{E}_I^{+}), 
\end{equation}
where $\mathcal{E}_I^{+} = \bigl\{\emb{i}^+ \mid i \in N_I \cup \{p\}\bigr\}$, and $\emb{i}^+=\texttt{MLP}^+(\textbf{z}^+_i)$ for each negative item $i \in N_\interactions$.\footnote{We tested negatives from $\items \setminus I^{+}$, but found no performance difference from $\items \setminus I$ in preliminary experiments; therefore, we reuse $N_I$ in Eq.~\eqref{eq:loss:rec:original} for efficiency.}
This contrasts $\emb{I}^{+}$ with $\emb{p}^{+}$ against the negatives.
We use the same $\tau_r$ for both losses for simplicity.
For efficiency, we compute $\emb{i}$ and $\emb{i}^{+}$ for all items in each batch (see Appendix~\ref{appendix:algorithms}).

\subsubsection{Cross-View Alignments}
The original view captures item consumption context but misses higher-order CF signals, while the augmented view provides such signals but may overemphasize them.
Thus, we align both views so that our model can effectively integrate their complementary information.
To this end, inspired by the CLIP-style contrastive loss~\cite{abs-2103-00020}, we adapt it to align the two views both at the interaction level and at the item level.

\smallsection{Interaction-level alignment}
We first align the original and augmented embeddings of interaction sets.  
Following the CLIP framework, each pair $(\emb{I}, \emb{I}^{+})$ is treated as a positive for $I\in\batch$, 
whereas other interaction embeddings in $\mathcal{B}$ are used as negatives, which improves efficiency and stability during training.
To achieve bidirectional consistency between $\emb{I}$ and $\emb{I}^{+}$, we design the interaction-level alignment loss in a symmetric manner as follows:

\vspace{-7mm}
\begin{align}
    \refstepcounter{equation}
    \tag{\small\theequation}\label{eq:loss:align:inter} \\[-0.5em]
    \mathcal{L}_{\text{inter-align}}(\batch) = -\frac{1}{2\card{\batch}}\sum_{I \in \batch}\Bigl(
        \underbrace{\ell^{(\tau_a)}\bigl(\emb{I}, \emb{I}^{+}; \mathcal{E}_{\batch}^{+}\bigr)}_{\text{original}\to\text{augmented}} + 
        \underbrace{\ell^{(\tau_a)}\bigl(\emb{I}^{+}, \emb{I}; \mathcal{E}_{\batch}\bigr)}_{\text{augmented}\to\text{original}}
    \Bigr), \nonumber
\end{align}
where $\mathcal{E}_{\batch} = \{\emb{I} \mid I \in \batch\}$ and $\mathcal{E}_{\batch}^{+} = \{\emb{I}^{+} \mid I \in \batch\}$, resp. 
For each $I$, the former encourages the original $\emb{I}$ to be closer to its augmented counterpart $\emb{I}^+$ against other augmented embeddings, while the latter enforces the opposite direction against other original ones.

\smallsection{Item-level alignment}
We then align the original and augmented embeddings of items.
In this step, we only consider positive items, as the goal is to align the same item pairs across both views, which can be sufficiently achieved using the positive items within $\mathcal{B}$, while including all other items increases computational cost.
Suppose $P_{\batch}$ is the set of positive items for $\batch$, i.e., $P_{\batch} = \{p \sim U(I) \mid I \in \batch\}$.
Then, the item-level alignment loss is defined as follows:

\vspace{-7mm}
\begin{align}
    \refstepcounter{equation}
    \tag{\small\theequation}\label{eq:loss:align:item} \\[-0.5em]
    \mathcal{L}_{\text{item-align}}(\batch) \!= \!-\frac{1}{2\card{P_{\batch}}}\!\sum_{p \in P_{\batch}}\Bigl(
        \underbrace{\ell^{(\tau_a)}\bigl(\emb{p}, \emb{p}^{+}; \mathcal{E}_{P_{\batch}}^{+}\bigr)}_{\text{original}\to\text{augmented}} + 
        \underbrace{\ell^{(\tau_a)}\bigl(\emb{p}^{+}, \emb{p}; \mathcal{E}_{P_{\batch}}\bigr)}_{\text{augmented}\to\text{original}}
    \Bigr),  \nonumber
\end{align}
where 
$\mathcal{E}_{P_{\batch}} = \{\emb{p} \mid p \in P_{\batch}\}$, and 
$\mathcal{E}_{P_{\batch}}^{+} = \{\emb{p}^{+} \mid p \in P_{\batch}\}$, resp.
The former aligns $\mathbf{e}_p$ to $\mathbf{e}_p^{+}$ against other augmented embeddings, 
while the latter aligns $\mathbf{e}_p^{+}$ back to $\mathbf{e}_p$ against other original embeddings.

\subsubsection{Final loss}
\label{sec:loss:final}
We describe the final loss $\mathcal{L}_\text{final}(\cdot)$ for training our model given $\batch \subseteq \traincol$, which is defined as follows:
{\small\begin{equation}
    \label{eq:loss:final}
    \mathcal{L}_\text{final}(\batch) = \mathcal{L}_\text{rec}(\batch) + \lambda_\text{align} \cdot \mathcal{L}_\text{align}(\batch),
\end{equation}}%
where $\mathcal{L}_\text{rec}(\cdot)$ is the recommendation loss, $\mathcal{L}_\text{align}(\cdot)$ is the alignment loss, and $\lambda_{\text{align}}$ is a hyperparameter that balances both components.
The loss functions $\mathcal{L}_\text{rec}(\cdot)$ and $\mathcal{L}_\text{align}(\cdot)$ are defined as follows:
\begin{align*}
    \mathcal{L}_\text{rec}(\batch) &=  \mathcal{L}_\text{InfoNCE}(\batch) 
+ \mathcal{L}_\text{InfoNCE}^+(\batch), \\
    \mathcal{L}_\text{align}(\batch) &= \mathcal{L}_\text{inter-align}(\batch) 
+ \mathcal{L}_\text{item-align}(\batch).
\end{align*}
Notice \method's parameters consist only of those of the MLP layers, i.e., $\Theta = \{\theta_{\texttt{MLP}}, \theta_{\texttt{MLP}^{+}}\}$,
where each is a two-layer network mapping from $\dlm$ to $d$ with a weight matrix and a bias per layer.
During training, for each $\mathcal{B}$ at every epoch, 
we compute $\mathcal{L}_\text{final}(\mathcal{B})$, 
backpropagate the gradients, and update  $\Theta$ via gradient descent.

\vspace{-1mm}
\subsection{Inference Phase (Algorithm~\ref{alg:inference})}
\label{sec:method:inference}
In the inference phase, we compute recommendation scores for all items with respect to a querying interaction set $I_q$, based on the trained parameters.
Our idea in this phase is to utilize the parameters $\theta_{\texttt{MLP}}$ in the original view.
The main reasons are: (1) higher-order CF signals from the augmented view have already been aligned into the original-view encoder $\texttt{MLP}(\cdot)$ during training, and (2) the augmentation of $I_q$ is therefore unnecessary at inference, leading to better efficiency.
For each item $i$, the recommendation score with respect to $I_q$ in Problem~\ref{prob} is defined as:
\begin{equation}
    \label{eq:inference}
    y_{\theta_{\texttt{MLP}}}(i \mid I_q) = \sigma\bigl(
        \emb{I_q}, \emb{i}
    \bigr),
\end{equation}
where 
$\sigma(\cdot, \cdot)$ denotes cosine similarity,
$\emb{I_q} = \texttt{MLP}\bigl(
        \texttt{mean}\bigl(
            \bigl\{\lr{i} \mid i \in I_q\bigr\}
        \bigr)
    \bigr)$, and $\emb{i} = \texttt{MLP}\bigl(\lr{i}\bigr)$.
Note that \textbf{our inference phase, as well as preprocessing and training, performs recommendations free from users, IDs, and GNNs}, 
as it efficiently computes the embedding of $I_q$ for an arbitrary user on-the-fly (\textit{\textbf{user-free}}), 
leverages item LRs (\textit{\textbf{ID-free}}), 
and exploits higher-order CF signals aligned into the original-view encoder without relying on any GNNs (\textit{\textbf{GNN-free}}).

\vspace{-2mm}
\subsection{Complexity Analysis}
Suppose $n = |\items|$, $m = \sum_{I \in \traincol}|I|$, and $h = |\traincol|$ denote the numbers of items, interactions, and interaction sets, respectively.
Let $\dlm$ be the dimension of LRs.
Other terms are hyperparameters. 
The detailed analysis and proofs are in Appendix~\ref{appendix:comp}.


\smallsection{Time complexities}
\method's time complexities are as follows:
\begin{theorem}[Time Complexities of \method]
\label{theorem:time}
In the preprocessing phase, \method takes $O(T_{\textnormal{LM}}n+ mn + (\dlm + \log{K_c})n^2)$ time, where $T_{\textnormal{LM}}$ is the time used by the LM.
Its training phase requires $O((T_i d n + K_c m + d h)\dlm + (n_s + |\batch|)hd)$ time for each epoch, where $T_i = h/|\batch|$ is the number of iterations over all batches.
Its inference phase takes $O(\dlm d n + n\log{K})$ time. \qedsymbol
\end{theorem}

Assuming the hyperparameters, including $\dlm$, are fixed constants, the preprocessing phase takes $O(mn + n^2)$, whereas the training phase takes $O(m + h + n)$.
Although the preprocessing phase does not exhibit linear scalability, it is performed only once, and can be further optimized using approximate $k$-NN (see Appendix~\ref{appendix:limitations}).
In contrast, the training phase scales linearly w.r.t. $m$, $h$, and $n$, and the inference phase also scales linearly with $n$. 

\begin{table*}[t]
\vspace{-3mm}
\setlength{\tabcolsep}{2.73pt}
\small
\caption{
Recommendation performance in terms of Recall@$K$ and NDCG@$K$ ($K=20$) under the user-wise holdout split setting.
We indicate whether each method satisfies three design principles: \texttt{UF} (user-free), \texttt{IF} (ID-free), and \texttt{GF} (GNN-free).
The best non-LR results are in red, and best LR results (excluding \method) are in blue.
Note that only \method adheres to all design principles, achieves the best accuracy, and successfully processes all datasets where some LR-based methods fail.
}
\label{table:overall}
\begin{threeparttable}
\centering
\begin{tabular}{c|ccc|ccccccc|ccccccc}

\hline
\toprule
          & \multicolumn{3}{c|}{\textbf{Principle}} & \multicolumn{7}{c|}{\textbf{Recall@20}}                                   & \multicolumn{7}{c}{\textbf{NDCG@20}} \\ 
\textbf{Model} & \multicolumn{1}{c}{\texttt{UF}} & \multicolumn{1}{c}{\texttt{IF}} & \multicolumn{1}{c|}{\texttt{GF}} & \movie   & \book    & \video   & \baby    & \steam  & \beauty  & \health  & \movie   & \book    & \video   & \baby    & \steam   & \beauty  &\health  \\ \midrule
MF-BPR    & \xmark  & \xmark  & \cmark  & 0.0580  & 0.0436  & 0.0177  & 0.0150  & 0.1610 & 0.0063  &  0.0091   & 0.0533  & 0.0387  & 0.0095  & 0.0074 &0.1537  & 0.0032 & 0.0048    \\ 
FISM-BPR  & \cmark  & \xmark  & \cmark  & 0.0861  & 0.0623  & 0.0392  & 0.0150  & 0.1801  &0.0079 &  0.0104  & 0.0801  & 0.0540  & 0.0220  & 0.0076  & 0.1431 & 0.0044 & 0.0060   \\ 
LightGCN  & \xmark   & \xmark   & \xmark   & 0.0860  & 0.0712  & 0.0732  & 0.0359  & 0.2013 & 0.0201 &   0.0250 & 0.0754  & 0.0597  & 0.0409  & 0.0185  & 0.1518 & 0.0110  & 0.0142   \\ 
XSimGCL   & \xmark   & \xmark  & \xmark  & \cellcolor{red!10}0.0967  & \cellcolor{red!10}0.0818  & \cellcolor{red!10}0.0897 &  \cellcolor{red!10}0.0390   &\cellcolor{red!10}0.2245  &\cellcolor{red!10}{0.0253} & \cellcolor{red!10}{0.0299}    &  \cellcolor{red!10}0.0866  & \cellcolor{red!10}0.0690  & \cellcolor{red!10}{0.0500}  & \cellcolor{red!10}0.0203  & \cellcolor{red!10}0.1664 & \cellcolor{red!10}{0.0140} & \cellcolor{red!10}{0.0170}   \\ \midrule
RLMRec    & \xmark  & \cmark  & \xmark  & 0.1046  & 0.0905  & o.o.t.     & o.o.t.     & o.o.t.     & o.o.t.   &  o.o.t.  & 0.0942  & 0.0741  & o.o.t.     & o.o.t.     & o.o.t.     & o.o.t.  &   o.o.t. \\ 
AlphaRec  & \xmark  & \cmark  & \xmark  & \cellcolor{blue!10}0.1219  & \cellcolor{blue!10}{0.0991}  & \cellcolor{blue!10}{0.1088}  & \cellcolor{blue!10}{0.0391}  & \cellcolor{blue!10}{0.2360}  & o.o.m.  & o.o.m. & \cellcolor{blue!10}{0.1141} & \cellcolor{blue!10}{0.0829}  & \cellcolor{blue!10}{0.0605}  & \cellcolor{blue!10}{0.0210}  & \cellcolor{blue!10}{0.1884}  & o.o.m.   &  o.o.m. \\ \midrule
\method & \cmark  & \cmark    & \cmark    & \textbf{0.1267}  & \textbf{0.1014}  & \textbf{0.1111}  & \textbf{0.0412}  & \textbf{0.2402}  & \textbf{0.0361}  & \textbf{0.0325} & \textbf{0.1194}  & \textbf{0.0861}  & \textbf{0.0615}  & \textbf{0.0219}  & \textbf{0.1938}  & \textbf{0.0200}    & \textbf{0.0184}  \\
 \midrule
\multicolumn{4}{c|}{\% inc. over best LR}   & 3.94\%$^{*}$  & 2.32\%$^{*}$  & 2.11\%$^{*}$  & 5.37\%$^{*}$  & 1.78\%$^{*}$  & - & - & 4.65\%$^{*}$  & 3.86\%$^{*}$  & 1.65\%$^{*}$  & 5.71\%$^{*}$  & 2.89\%$^{*}$  & - & - \\ 
\multicolumn{4}{c|}{\% inc. over best non-LR} &  31.00\%$^{*}$  & 23.96\%$^{*}$  & 23.81\%$^{*}$  & 5.58\%$^{*}$  & 6.99\%$^{*}$  & 42.69\%$^{*}$  & 8.70\%$^{*}$ &  37.90\%$^{*}$  & 24.78\%$^{*}$  & 23.07\%$^{*}$  & 9.26\%$^{*}$  & 16.49\%$^{*}$ & 42.86\%$^{*}$ & 8.24\%$^{*}$    \\   
\bottomrule\hline 
\end{tabular}
\begin{tablenotes}[flushleft]
    {\item[] - Out-of-time (o.o.t.) indicates that the total time including preprocessing exceeded 50 hours, while out-of-memory (o.o.m.) denotes GPU memory overflow.}
    {\item[] - Percentage increases marked with $*$ are statistically significant ($p$-value $\leq 0.05$).}
    \end{tablenotes}
\end{threeparttable}
\end{table*}

\smallsection{Space complexities} \method's space complexities are as:
\begin{theorem}[Space Complexities of \method]
The preprocessing phase uses $O(\dlm n + K_c(m+n))$ space.
The training phase requires $O(m + d(n + \dlm + n_s + |\batch|))$ space.
The inference phase takes $O((\dlm + d)n + K)$ space. \qedsymbol
\label{theorem:space}
\end{theorem}

The preprocessing and training phases use $O(m + n)$ space, while the inference phase uses $O(n)$ space when the hyperparameters, including $\dlm$, are fixed constants.
Note that LightGCN and AlphaRec use $O(n + h + m)$ space for training and inference, where $h$ and $m$ denote the numbers of users and edges, while \method reduces $O(h)$ space in training and $O(h + m)$ in inference.

\section{Experiments}
\label{sec:exp}
We present experimental results to evaluate the effectiveness of \method compared with state-of-the-art baselines for Problem~\ref{prob}.

\vspace{-2mm}
\subsection{Experimental Settings}

\smallsection{Datasets and competitors}
We conducted experiments on seven real-world datasets in Table~\ref{tab:datasets}, where details of each dataset are in Appendix~\ref{appendix:settings}.
We group the baselines of \method as follows:
\begin{itemize}[leftmargin=4mm]
    \item {
        \textbf{Non-LR-based methods.} We considered MF-BPR~\cite{KorenBV09}, a widely used  baseline that stores user embeddings, and FISM-BPR~\cite{KabburNK13}, a well-known user-free model.
        We included LightGCN~\cite{0001DWLZ020} and XSimGCL~\cite{YuXCCHY24}, which are popular GNN-based approaches.
        All of these methods are ID-dependent, not using LRs.
    }
    \item{
        \textbf{LR-based methods.} We included RLMRec~\cite{RenWXSCWY024} and AlphaRec \cite{Sheng0ZCWC25}, which are state-of-the-art ID-free approaches. 
        Details of the LMs used in this section are  in Appendix~\ref{appendix:settings}, with additional analyses on the effect of different LMs provided in Appendix~\ref{appendix:Effect_of_LM}.
    }
\end{itemize}

\begin{table}
\centering
\caption{
Statistics of real-world recommendation datasets. 
\label{tab:datasets}
}
\resizebox{1.015\linewidth}{!}{%
\small
\setlength{\tabcolsep}{1pt}
\hspace{-1.15mm}
\begin{tabular}{lrrrrrrr}
\hline
\toprule

\textbf{Datasets} & \movie & \book & \video & \baby & \steam & \beauty &\health  \\ 
\midrule
\bf \#Users & 26,073   & 71,306   & 94,762   & 150,777  & 334,730 & 729,576 & 796,054   \\
\bf \#Items & 12,464   & 40,523   & 25,612   & 36,013   & 15,068   & 207,649 & 184,346   \\
\bf \#Inter. & 875,906  & 2,206,865 & 814,586 & 1,241,083 & 4,216,781  & 6,624,441 & 7,176,552 \\

\bottomrule 
\hline
\end{tabular}
}
\end{table}

\smallsection{Training and evaluation protocol}
Each dataset was divided into training, validation, and test sets, where the detailed ratios and splitting strategies are provided in the following subsections.
We tuned the hyperparameters of each model, including \method, on the validation set using Recall@20.
Using the selected hyperparameters, we then evaluated the models on the test set in terms of Recall@20 and NDCG@20.
We ran each experiment five times with different random seeds and reported the average performance.
Refer to Appendix~\ref{appendix:settings} for details of the settings.

\smallsection{Reproducibility} The code and datasets used in the paper are available at \url{https://github.com/minseojeonn/AlphaFree}.

\vspace{-2mm}
\subsection{Recommendation Performance}
\label{sec:exp:overall}

We evaluated the performance of top-$K$ recommendation following recent studies~\cite{0003CSWC24,Sheng0ZCWC25}, using a user-wise holdout split. 
Specifically, we randomly divided each user’s interaction history into training, validation, and test sets with a ratio of 4:3:3.
The results are reported in Table~\ref{table:overall}, and we observe the following findings:
\begin{itemize}[leftmargin=4mm]
    \item{
        \method outperforms LR-based methods with gains of up to 5.37\% (Recall@20) and 5.71\% (NDCG@20), and significantly outperforms non-LR-based ones with improvements of up to 42\% on both metrics, while uniquely satisfying all design principles.
    }
    \item{
        Among LR-based methods, only \method successfully handles all datasets, whereas others fail on large ones, with RLMRec suffering from out-of-time even on smaller datasets and AlphaRec running out-of-memory on larger ones such as \beauty and \health, highlighting the advantage of being user-free.
    }
    \item{
        \method benefits from being ID-free, as ID-dependent methods such as MF-BPR and FISM-BPR yield  lower accuracy, indicating that LR-based (ID-free) approaches, including \method, provide richer semantics and better performance.
    }
    \item{
        By being GNN-free, \method consistently outperforms GNN-dependent methods, including LightGCN and AlphaRec, in both accuracy and efficiency.
    }
\end{itemize}
\vspace{-3mm}

\begin{figure}[t!]
    \vspace{-2mm}
    \includegraphics[width=0.7\linewidth]{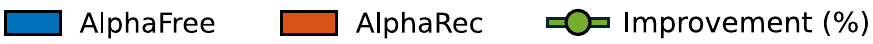}
    \\
    \centering
    \begin{subfigure}[b]{0.317\linewidth}
        \includegraphics[width=\linewidth]{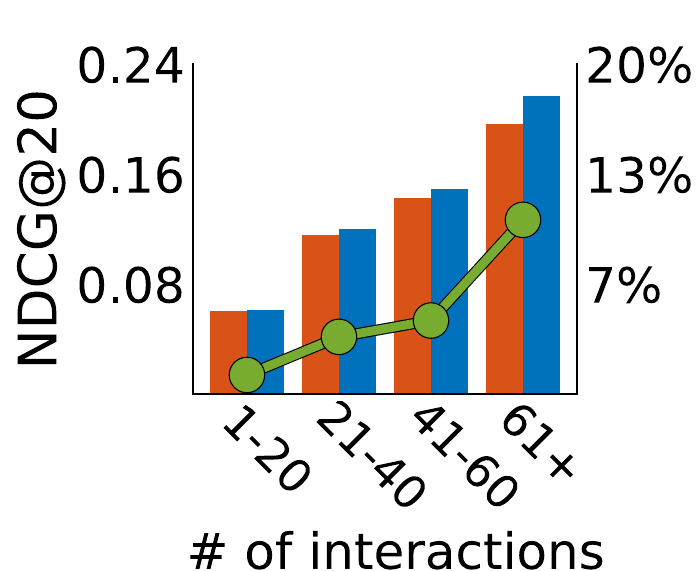}
        \caption{\video}
    \end{subfigure}
    \hfill
    \begin{subfigure}[b]{0.30\linewidth}
        \includegraphics[width=\linewidth]{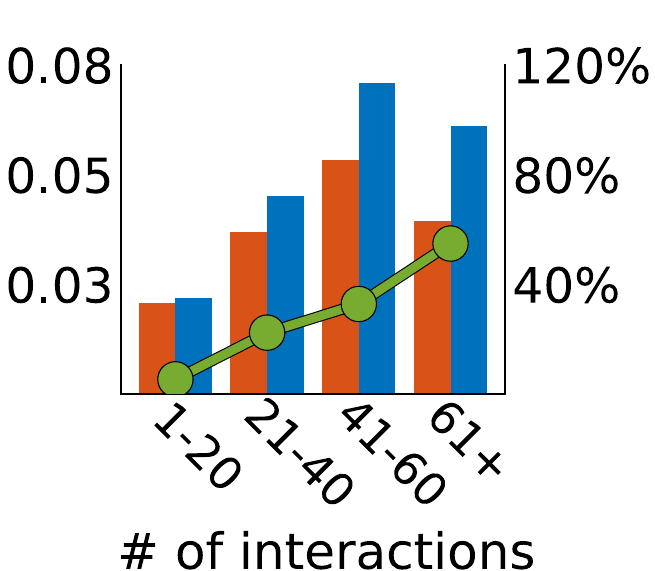}
        \caption{\baby}
    \end{subfigure}
    \hfill
    \begin{subfigure}[b]{0.317\linewidth}
        \includegraphics[width=\linewidth]{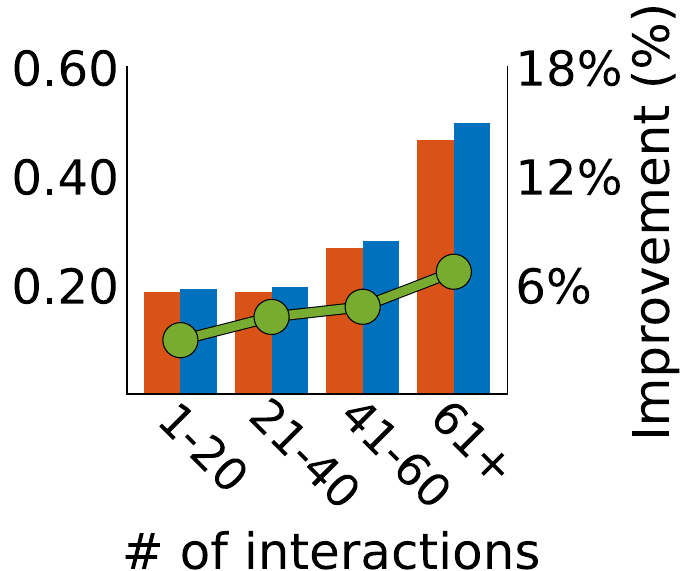}
        \caption{\steam}
    \end{subfigure}
    \caption{
        Performance of \method and AlphaRec in NDCG@20 across groups with different numbers of interactions. 
        The {green} line indicates the relative improvement of \method over AlphaRec.
        \method shows consistent improvements, which grow with more interactions.
    }
    \label{fig:exp:inter_groups}
\end{figure}

\subsection{Performance across Interaction Groups}
We analyzed performance across interaction groups to investigate the effect of being GNN-free.
To this end, we compared \method with AlphaRec, the second-best  model (Table~\ref{table:overall}), which also utilizes item LRs but still relies on GNNs.
We classified each $I \in \traincol$ into one of four groups (i.e., 1–20, 21–40, 41–60, and 61+) based on $\card{I}$ and measured NDCG@20 for each group. 
Figure~\ref{fig:exp:inter_groups} shows the results, with the percentage increase of ours over AlphaRec shown on the right y-axis.
\method consistently outperforms AlphaRec across all groups, and the relative improvement becomes larger as $|I|$ increases.
In particular, the improvement in the 61+ group (or heavy users) is more pronounced than in the other groups, indicating that AlphaRec using GNNs shows limited performance for those users as their representations become excessively averaged (i.e., over-smoothed), whereas our GNN-free design avoids such an issue.

\vspace{-2mm}
\subsection{Generalization to Unseen Interaction Sets}
\label{sec:exp:generalization_unseen_interaction}
We evaluated the generalization performance of \method on interaction sets unseen during training, i.e., in a user cold-start setting\footnote{
In this experimental setup, each new user was given on average 1.09–2.95 interactions as a query at test time (see Appendix~\ref{appendix:settings} for details).
}.
For each dataset, we first split the users into training, validation, and test sets with an 8:1:1 ratio.
For evaluation, we sampled 10\% of interactions (at least one) from each user’s set $I$ in the validation and test data to form $I_q$, and used the remaining ones as ground truth.
%
From Table~\ref{tab:unseen}, we observe the following:
\begin{itemize}[leftmargin=4mm]
    \item{
        \method shows consistently superior generalization performance compared to the tested baselines, achieving improvements of up to 38.65\% in Recall@20 and 35.95\% in NDCG@20, respectively.
    }
    \item{
        As a user-dependent method, AlphaRec shows limited generalization compared to \method, indicating that it struggles to learn effective embeddings for new users.
    }
    \item{
        For ID-dependent methods, the user-free FISM-BPR performs better than LightGCN, whereas \method further improves performance by leveraging LRs. 
    }
\end{itemize}
These verify that user-free architectures generalize more effectively to new users with few interactions, while leveraging LRs provides additional gains beyond ID- and user-dependent models.

\
\begin{table}[t!]
\caption{
Performance on interaction sets of new users unseen during training (best in bold; second-best underlined).
\method outperforms the tested baselines by effectively recommending for new users. 
}
\label{tab:unseen}
\small
\setlength{\tabcolsep}{1.3pt}
\begin{tabular}{c|rrrr|rrrr}
\hline
\toprule
\multirow{2}{*}{\textbf{Model}} & \multicolumn{4}{c|}{\textbf{Recall@20}} & \multicolumn{4}{c}{\textbf{NDCG@20}} \\
                         & \movie & \book &\video & \baby &  \movie & \book  & \video & \baby \\ \midrule
FISM-BPR               & 0.0570 &0.0494&	0.0601&	0.0135&	0.0853&	0.0601&	0.0397&	0.0090
\\
LightGCN                &0.0401 & 0.0302&	0.0438&	0.0125&	0.0614	&0.0423	&0.0274	&0.0076 
\\      
AlphaRec                & \underline{0.0681} & \underline{0.0621} & \underline{0.0732} & \underline{0.0216} & \underline{0.0980} & \underline{0.0837} & \underline{0.0474} & \underline{0.0138} \\
\midrule
\method                 & \textbf{0.0706} & \textbf{0.0646} & \textbf{0.0922} & \textbf{0.0300} & \textbf{0.1027} & \textbf{0.0850} & \textbf{0.0596} & \textbf{0.0188} \\ \midrule
\% inc.                & 3.63\% & 3.94\% & 25.90\% & 38.65\% & 4.89\% & 1.51\% & 25.71\% & 35.95\% \\ 
\bottomrule
\hline
\end{tabular}
\end{table}

\begin{figure}[t]
    \centering
    \includegraphics[width=0.89\linewidth]{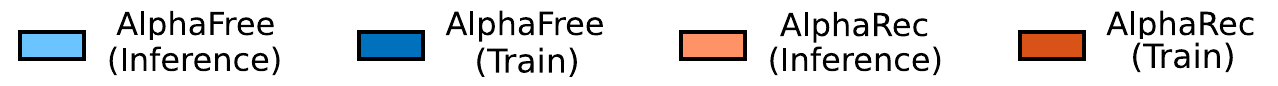}   
    \includegraphics[width=0.95\linewidth]{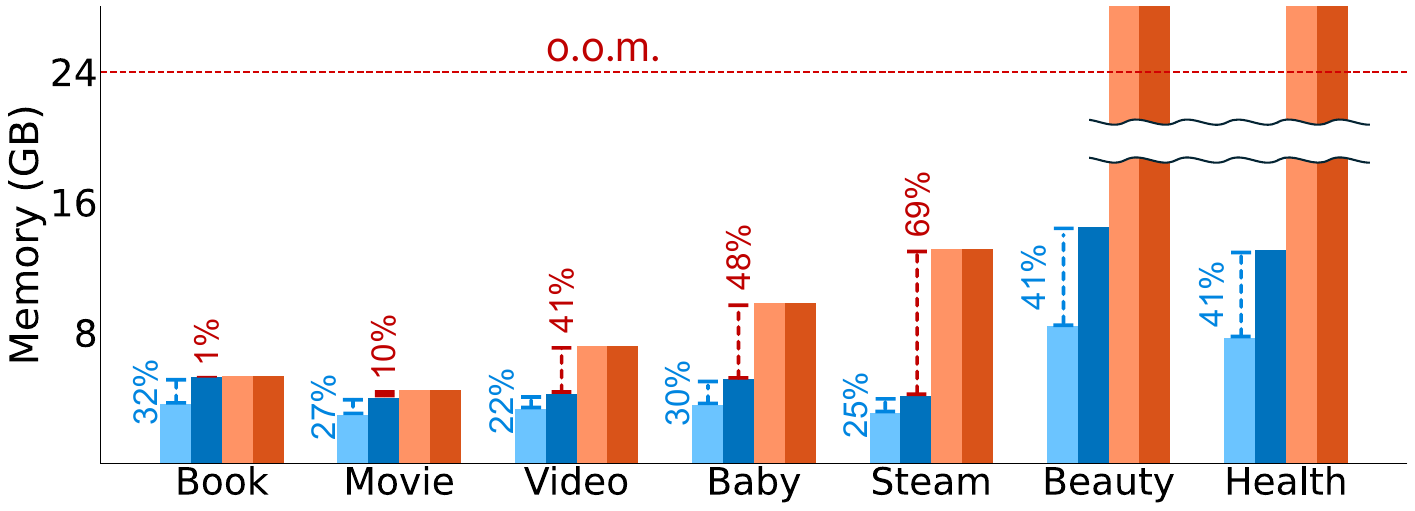}
    \caption{
        GPU memory usage of \method and AlphaRec during training and inference, showing \method consistently requires less memory, and AlphaRec runs out of memory (o.o.m.) on large datasets such as \beauty and \health.
    }
    \label{fig:exper:vram_usage_LR_compare}
    \vspace{2mm}
\end{figure}

\begin{figure}[t]
    \centering
    \includegraphics[width=0.9\linewidth]{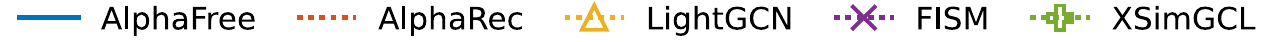}
    \centering
    \begin{subfigure}[b]{0.49\linewidth}
        \includegraphics[width=0.9\linewidth]{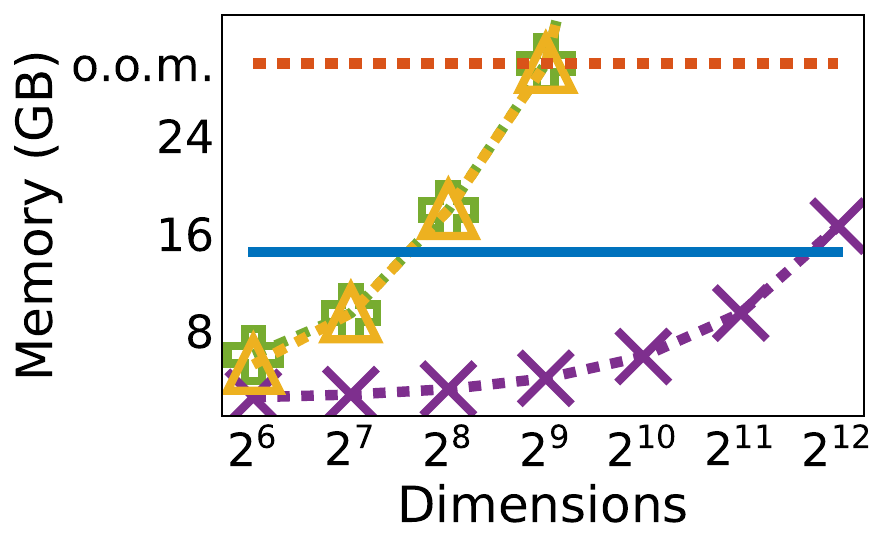}
        \caption{Train}
    \end{subfigure}
    \begin{subfigure}[b]{0.49\linewidth}
        \includegraphics[width=0.9\linewidth]{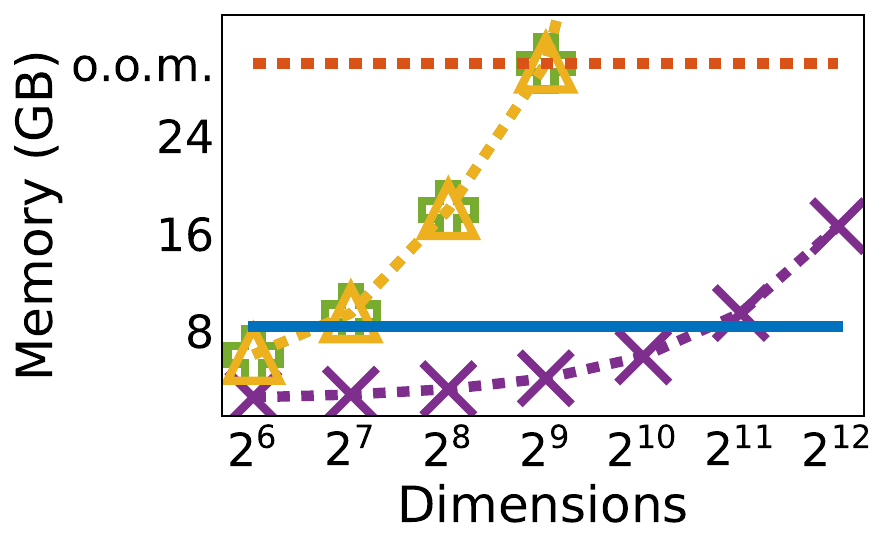}
        \caption{Inference}
    \end{subfigure}
    \caption{
        GPU memory usage on \beauty under varying embedding dimensions. For \method and AlphaRec, the dimension is fixed by the language model ($\dlm=2^{12}$), where AlphaRec runs out-of-memory. Other baselines require less memory at smaller dimensions but grow rapidly, eventually surpassing \method after $\dlm$.
    } \label{fig:exp:usage:dim}
\end{figure}

\vspace{-4.5mm}
\subsection{Space Efficiency}

We analyze the space efficiency of \method gained from its user-free design.
The runtime analysis is provided in Appendix~\ref{appendix:runtime}.

\smallsection{Memory usage over LR-based methods}
For LR-based models, language models (LMs) produce LRs of a fixed size $\dlm$ that cannot be easily adjusted, which becomes a major  bottleneck in space.
Thus, we analyzed the average GPU memory usage per epoch of \method and AlphaRec during both training and inference, where they used the same LRs.
As shown in Figure~\ref{fig:exper:vram_usage_LR_compare}, \method uses less memory than AlphaRec, which fails to run on larger datasets such as \beauty and \health.
In particular, \method reduces the usage by up to 69\% compared to AlphaRec during training.
This results from the user-free design of \method (i.e., not storing user embeddings), unlike AlphaRec.
We also observe that \method requires up to 41\% less memory for inference compared to training, as the inference phase uses only the original-view encoder.

\smallsection{Memory usage w.r.t. embedding dimensions}
Unlike the above LR-based models, the embedding dimension $d$ of traditional methods can be easily adjusted.
For \method and AlphaRec, $\dlm = 2^{12}$ was fixed by the LM,
while we varied $d$ for the other methods from $2^6$ to $2^{12}$ to examine their memory usage on the \beauty dataset.
As shown in Figure~\ref{fig:exp:usage:dim}, they use less GPU memory at smaller $d$, but their usage increases as $d$ grows.
In particular, LightGCN and XSimGCL (i.e., user-dependent) show steeper growth than FISM-BPR (i.e., user-free).
\method uses more memory than those methods with small $d$ due to high-dimensional LRs, but it achieves much higher accuracy (Section~\ref{sec:exp:overall}), whereas those methods overfit at large dimension $d$.

\vspace{-1mm}
\subsection{Ablation Study}
We further analyze the contribution of each component through ablation studies: 
\textbf{w/o semantic filtering} disables the semantic filtering module used for selecting similar items;
\textbf{w/o inter-align} removes the interaction-level alignment loss 
$\mathcal{L}_\text{inter-align}$; 
\textbf{w/o item-align} removes the item-level alignment loss 
$\mathcal{L}_\text{item-align}$;  and 
\textbf{w/o all-align} removes all alignment losses by setting $\lambda_\text{align}=0$, and uses only $\mathcal{L}_\text{InfoNCE}$.
Figure~\ref{fig:ablation} shows the results of our ablation studies, where removing any component leads to performance degradation, verifying that each module contributes positively to \method with varying impacts across datasets.
Note that the full version of \method consistently achieves the best performance.

\subsection{Effect of Hyperparameters}
\label{sec:hyper_sens}
We analyzed the effect of $K_c$ and $\lambda_\text{align}$, which controls the number of similar items and the strength of the cross-view alignment.
Since data density can affect augmentation, we selected \movie\ and \baby, whose densities are 0.27\% and 0.02\%, to examine the effects on datasets with different densities.

\smallsection{Effect of $K_c$} 
As shown in Figure~\ref{fig:hyperparam_sense}, \method benefits from a larger $K_c$ in the sparser \baby\ dataset, whereas \movie\ achieves the best performance around $K_c=5$, indicating that a larger $K_c$ is not good for denser datasets. 

\smallsection{Effect of $\lambda_\text{align}$}
In the denser \movie\ dataset, larger $\lambda_\text{align}$ improves performance by enforcing consistency between the two views.
In the sparser \baby\ dataset, too strong alignment degrades performance, as the augmentations in sparse data may contain noise; in this case, a weaker alignment can be helpful.

\begin{figure}[t]
    \centering
    \includegraphics[width=0.85\linewidth]{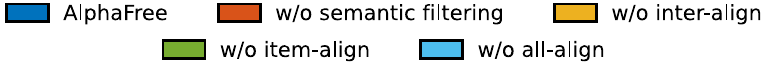}
    \\
    \begin{subfigure}[b]{0.3\linewidth}
        \includegraphics[width=\linewidth]{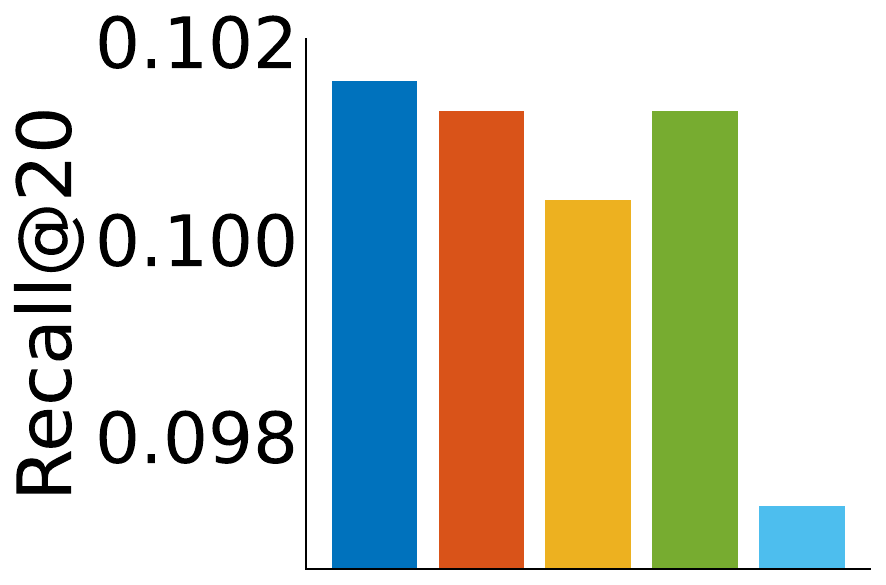}
        \caption{\book}
    \end{subfigure}
    \begin{subfigure}[b]{0.3\linewidth}
        \includegraphics[width=\linewidth]{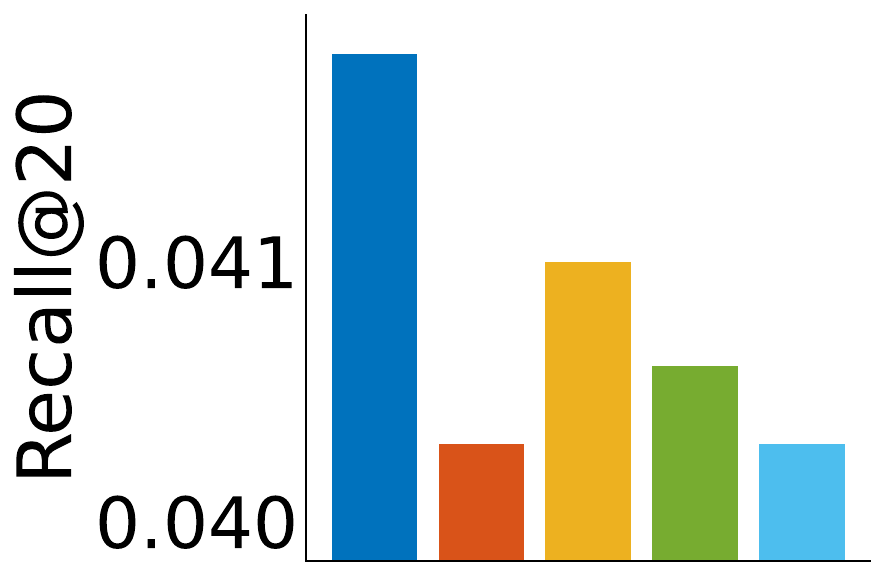}
        \caption{\baby}
    \end{subfigure}
    \begin{subfigure}[b]{0.3\linewidth}
        \includegraphics[width=\linewidth]{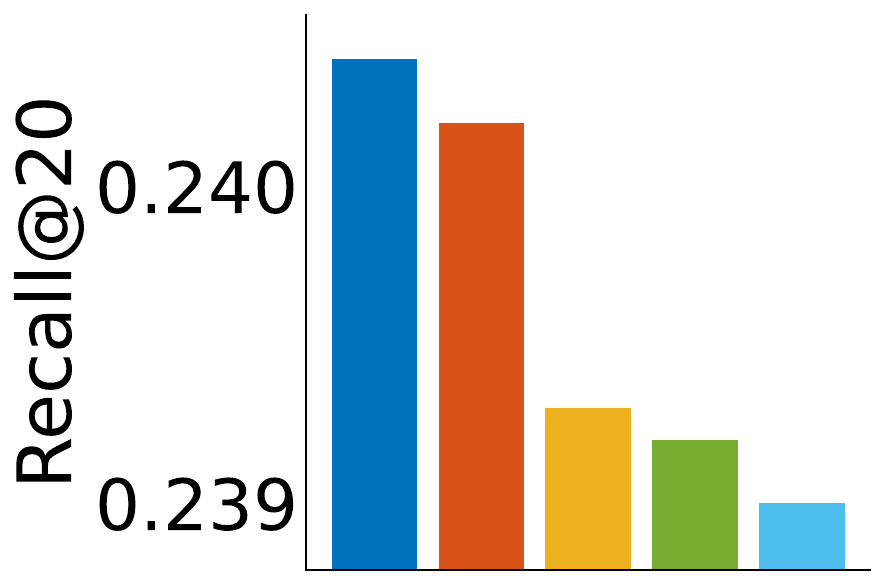}
        \caption{\steam}
    \end{subfigure}
    \caption{
    Ablation results showing that each module contributes positively, with varying effects across datasets, and the full model (\method) performs the best.}
    \label{fig:ablation}
    \vspace{2mm}
\end{figure}

\begin{figure}[t]
    \centering
    \begin{subfigure}{0.49\linewidth}
        \centering
        \includegraphics[width=\linewidth]{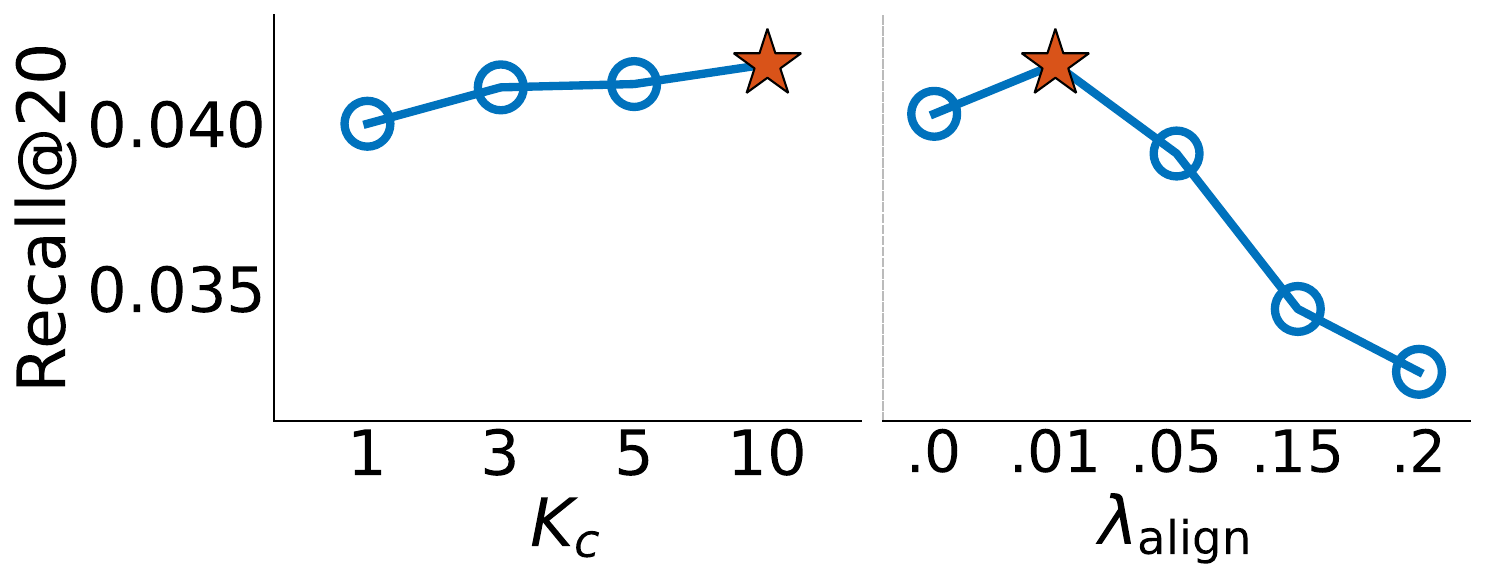}
        \caption{\baby}
        \label{fig:hyper_sense:a}
    \end{subfigure}
    \begin{subfigure}{0.49\linewidth}
        \centering
        \includegraphics[width=\linewidth]{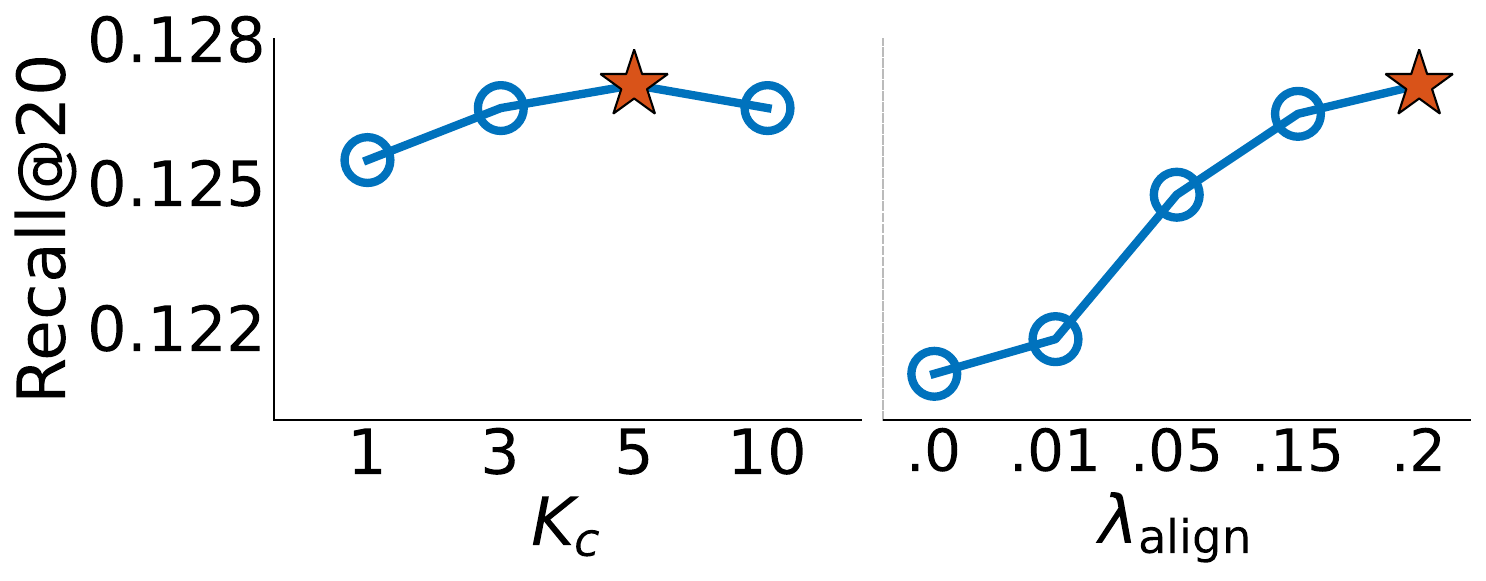}
        \caption{\movie}
        \label{fig:hyper_sense:b}
    \end{subfigure}
    \hfill 
    \caption{
        Effects of hyperparameters $K_c$ and $\lambda_{\textnormal{align}}$, where the best points are marked with a star symbol.
        The trends differ across datasets, suggesting that proper tuning of them is required to adapt to various data characteristics.
    } 
    \label{fig:hyperparam_sense}
\end{figure}

\subsection{Qualitative Case Study}
\label{sec:exp:case}
\begin{figure}
    \centering
    \includegraphics[width=\linewidth]{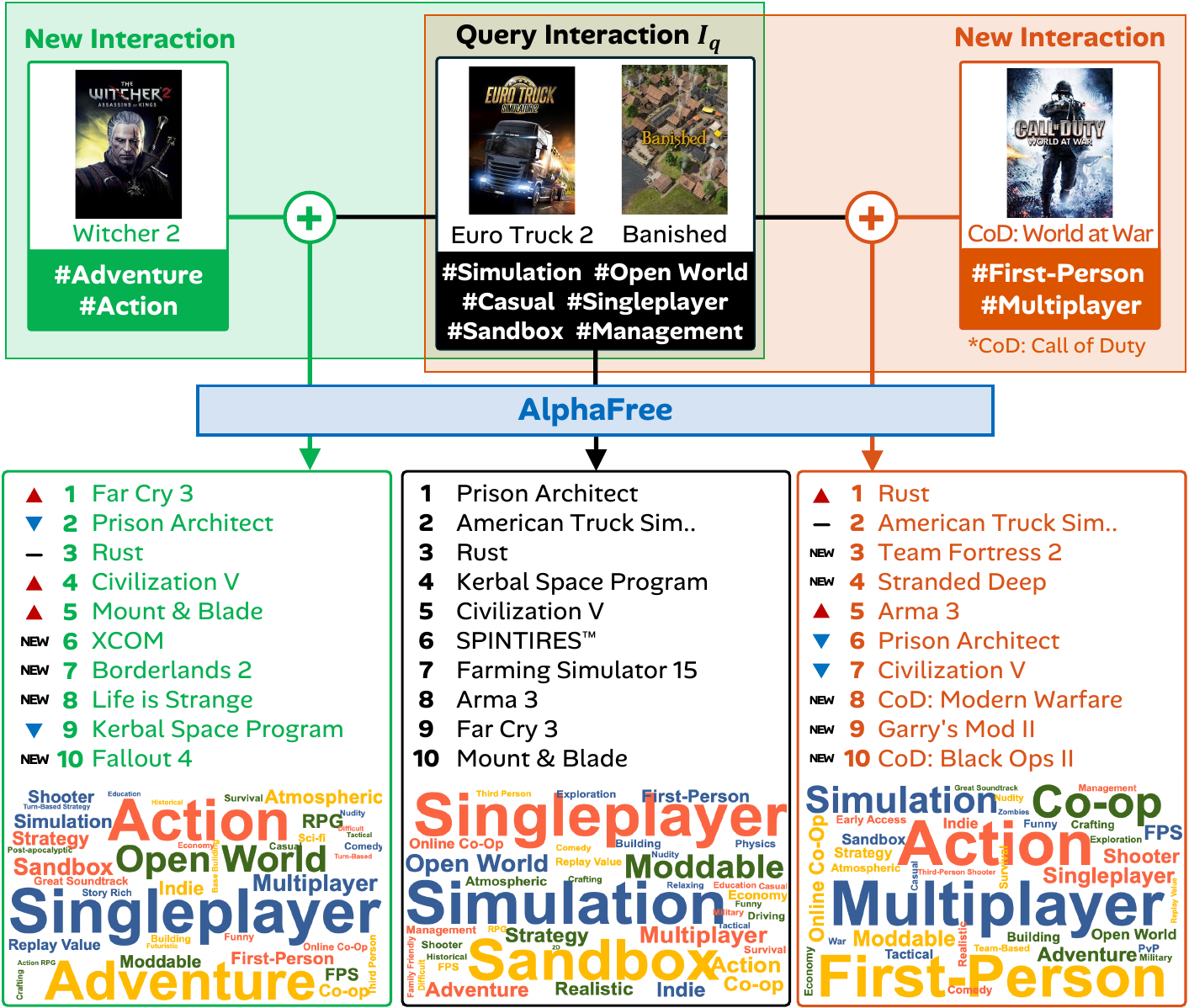}
    \caption{
        Case study on the \steam dataset, where the middle shows the results for unseen $I_q$, and the left and right show updates after adding a new item into $I_q$ with word clouds of user tags for recommended games.
        \label{fig:exp:case}
    }
\end{figure}

We conducted a case study to qualitatively verify the effectiveness of the user-free and ID-free designs of \method.
To this end, we trained \method on $\traincol$ of the \steam dataset and examined how it adapts to new interactions.
Specifically, we set an unseen $I_q = \{\textsc{EuroTruck}, \textsc{Banished}\}$, whose context is characterized by several tags such as \texttt{\#Simulation} and \texttt{\#Sandbox}. 
The resulting recommendations of \method for $I_q$ are shown in the middle of Figure~\ref{fig:exp:case}, where these tags are dominant, indicating \method  effectively captures the context of $I_q$.
Next, adding \textsc{Witcher2} (\texttt{\#Action} and \texttt{\#Adventure}) to $I_q$ and feeding it to \method without retraining yields the left recommendation list, where action games rise or newly appear, and those tags become dominant.
In contrast, when adding \textsc{CoD} to $I_q$, \method yields the right result, where FPS games rise or newly appear, and the \texttt{\#First-Person} and \texttt{\#Multiplayer} tags become dominant.
These results show that \method effectively captures query contexts (ID-free) and adapts to new interactions (user-free).

\vspace{-2mm}
\section{Conclusion}
\label{sec:conclusion}
In this paper, we proposed \method, a novel top-$K$ recommendation framework that is free from users, IDs, and GNNs.
Throughout this work, we demonstrate that even a simple and lightweight model, such as MLPs, can achieve strong recommendation performance by removing dependencies on users, IDs, and GNNs.
Specifically, \method augments each item and interaction set with behaviorally and semantically similar items, and contrastively aligns the original and augmented views during training, enabling the model to leverage high-order CF signals without relying on GNNs.
At inference, only the original-view encoder consisting of a lightweight MLP is used, providing on-the-fly recommendations for a given query, while remaining free from users, IDs, and GNNs.
Extensive experiments on various benchmarks and settings demonstrate that \method outperforms state-of-the-art methods, achieving up to around 40\% higher accuracy over non-LR-based methods and 5.7\% higher accuracy over LR-based methods, while reducing GPU memory usage by up to 69\%. 
We discuss the limitations and future directions of \method in Appendix~\ref{appendix:limitations}.

\vspace{-2mm}
\section*{Acknowledgements}
We would like to thank all anonymous reviewers for their dedicated reviews and comments.
This work was supported by the Institute of Information \& Communications Technology Planning \& Evaluation (IITP) - Innovative Human Resource Development for Local Intellectualization program grant funded by the Korea government (MSIT) (IITP-2026-RS-2022-00156360, 50\%) and 
the Convergence security core talent training business support program (IITP-2026-RS-2024-00426853, 50\%).

\bibliographystyle{ACM-Reference-Format}
\bibliography{myref-short}

@inproceedings{KabburNK13,
  author       = {Santosh Kabbur and
                  Xia Ning and
                  George Karypis},
  title        = {{FISM:} factored item similarity models for top-N recommender systems},
  booktitle    = {SIGKDD},
  year         = {2013}
}

@inproceedings{0001DWLZ020,
  author       = {Xiangnan He and
                  Kuan Deng and
                  Xiang Wang and
                  Yan Li and
                  Yong{-}Dong Zhang and
                  Meng Wang},
  title        = {LightGCN: Simplifying and Powering Graph Convolution Network for Recommendation},
  booktitle    = {SIGIR},
  year         = {2020}
}

@article{YuXCCHY24,
  author       = {Junliang Yu and
                  Xin Xia and
                  Tong Chen and
                  Lizhen Cui and
                  Nguyen Quoc Viet Hung and
                  Hongzhi Yin},
  title        = {XSimGCL: Towards Extremely Simple Graph Contrastive Learning for Recommendation},
  journal      = {{IEEE} Trans. Knowl. Data Eng.},
  volume       = {36},
  number       = {2},
  pages        = {913--926},
  year         = {2024}
}

@inproceedings{RenWXSCWY024,
  author       = {Xubin Ren and
                  Wei Wei and
                  Lianghao Xia and
                  Lixin Su and
                  Suqi Cheng and
                  Junfeng Wang and
                  Dawei Yin and
                  Chao Huang},
  title        = {Representation Learning with Large Language Models for Recommendation},
  booktitle    = {WWW},
  year         = {2024}
}

@inproceedings{Sheng0ZCWC25,
  author       = {Leheng Sheng and
                  An Zhang and
                  Yi Zhang and
                  Yuxin Chen and
                  Xiang Wang and
                  Tat{-}Seng Chua},
  title        = {Language Representations Can be What Recommenders Need: Findings and
                  Potentials},
  booktitle    = {ICLR},
  year         = {2025}
}

@inproceedings{HeLZNHC17,
  author       = {Xiangnan He and
                  Lizi Liao and
                  Hanwang Zhang and
                  Liqiang Nie and
                  Xia Hu and
                  Tat{-}Seng Chua},
  title        = {Neural Collaborative Filtering},
  booktitle    = {WWW},
  year         = {2017}
}

@inproceedings{Rendle10,
  author       = {Steffen Rendle},
  editor       = {Geoffrey I. Webb and
                  Bing Liu and
                  Chengqi Zhang and
                  Dimitrios Gunopulos and
                  Xindong Wu},
  title        = {Factorization Machines},
  booktitle    = {ICDM},
  year         = {2010}
}

@inproceedings{RendleFGS09,
  author       = {Steffen Rendle and
                  Christoph Freudenthaler and
                  Zeno Gantner and
                  Lars Schmidt{-}Thieme},
  title        = {{BPR:} Bayesian Personalized Ranking from Implicit Feedback},
  booktitle    = {UAI},
  year         = {2009}
}

@article{abs-2407-13699,
  author       = {Shaina Raza and
                  Mizanur Rahman and
                  Safiullah Kamawal and
                  Armin Toroghi and
                  Ananya Raval and
                  Farshad Navah and
                  Amirmohammad Kazemeini},
  title        = {A Comprehensive Review of Recommender Systems: Transitioning from
                  Theory to Practice},
  journal      = {CoRR},
  year         = {2024},
  volume       = {abs/2407.13699}
}

@inproceedings{LiHW18,
  author       = {Qimai Li and
                  Zhichao Han and
                  Xiao{-}Ming Wu},
  title        = {Deeper Insights Into Graph Convolutional Networks for Semi-Supervised
                  Learning},
  booktitle    = {AAAI},
  year         = {2018},
}

@inproceedings{NingK11,
  author       = {Xia Ning and
                  George Karypis},
  title        = {{SLIM:} Sparse Linear Methods for Top-N Recommender Systems},
  booktitle    = {ICDM},
  year         = {2011}
}

@article{abs-2407-21783,
  author = {Abhinav Dubey and others},
  title        = {The Llama 3 Herd of Models},
  journal      = {CoRR},
  year         = {2024},
  volume       = {abs/2407.21783}
}

@article{KorenBV09,
  author       = {Yehuda Koren and
                  Robert M. Bell and
                  Chris Volinsky},
  title        = {Matrix Factorization Techniques for Recommender Systems},
  journal      = {Computer},
  year         = {2009},
  volume       = {42},
  number       = {8},
  pages        = {30--37}
}

@article{abs-2201-10005,
  author       = {Arvind Neelakantan and others},
  title        = {Text and Code Embeddings by Contrastive Pre-Training},
  journal      = {CoRR},
  volume       = {abs/2201.10005},
  year         = {2022},
  volume       = {abs/2201.10005}
}

@inproceedings{Geng0FGZ22,
  author       = {Shijie Geng and
                  Shuchang Liu and
                  Zuohui Fu and
                  Yingqiang Ge and
                  Yongfeng Zhang},
  title        = {Recommendation as Language Processing {(RLP):} {A} Unified Pretrain,
                  Personalized Prompt {\&} Predict Paradigm {(P5)}},
  booktitle    = {RecSys},
  year         = {2022}
}

@inproceedings{BaoZZWF023,
  author       = {Keqin Bao and
                  Jizhi Zhang and
                  Yang Zhang and
                  Wenjie Wang and
                  Fuli Feng and
                  Xiangnan He},
  title        = {TALLRec: An Effective and Efficient Tuning Framework to Align Large
                  Language Model with Recommendation},
  booktitle    = {RecSys},
  year         = {2023}
}

@inproceedings{Wang0WFC19,
  author       = {Xiang Wang and
                  Xiangnan He and
                  Meng Wang and
                  Fuli Feng and
                  Tat{-}Seng Chua},
  title        = {Neural Graph Collaborative Filtering},
  booktitle    = {SIGIR},
  year         = {2019}
}

@inproceedings{YuY00CN22,
  author       = {Junliang Yu and
                  Hongzhi Yin and
                  Xin Xia and
                  Tong Chen and
                  Lizhen Cui and
                  Quoc Viet Hung Nguyen},
  title        = {Are Graph Augmentations Necessary?: Simple Graph Contrastive Learning
                  for Recommendation},
  booktitle    = {SIGIR},
  year         = {2022}
}

@inproceedings{JuanZCL16,
  author       = {Yu{-}Chin Juan and
                  Yong Zhuang and
                  Wei{-}Sheng Chin and
                  Chih{-}Jen Lin},
  title        = {Field-aware Factorization Machines for {CTR} Prediction},
  booktitle    = {RecSys},
  year         = {2016}
}

@inproceedings{Steck19,
  author       = {Harald Steck},
  title        = {Embarrassingly Shallow Autoencoders for Sparse Data},
  booktitle    = {WWW},
  year         = {2019}
}

@article{abs-2312-10073,
  author       = {Zefeng Chen and
                  Wensheng Gan and
                  Jiayang Wu and
                  Kaixia Hu and
                  Hong Lin},
  title        = {Data Scarcity in Recommendation Systems: {A} Survey},
  journal      = {CoRR},
  year         = {2023},
  volume       = {abs/2312.10073}
}

@inproceedings{Gao0HCZFZ22,
  author       = {Yunjun Gao and
                  Yuntao Du and
                  Yujia Hu and
                  Lu Chen and
                  Xinjun Zhu and
                  Ziquan Fang and
                  Baihua Zheng},
  title        = {Self-Guided Learning to Denoise for Robust Recommendation},
  booktitle    = {SIGIR},
  year         = {2022}
}

@inproceedings{HuangDDYFW021,
  author       = {Tinglin Huang and
                  Yuxiao Dong and
                  Ming Ding and
                  Zhen Yang and
                  Wenzheng Feng and
                  Xinyu Wang and
                  Jie Tang},
  title        = {MixGCF: An Improved Training Method for Graph Neural Network-based
                  Recommender Systems},
  booktitle    = {KDD},
  year         = {2021}
}

@article{abs-2303-10993,
  author       = {T. Konstantin Rusch and
                  Michael M. Bronstein and
                  Siddhartha Mishra},
  title        = {A Survey on Oversmoothing in Graph Neural Networks},
  journal      = {CoRR},
  year         = {2023},
  volume       = {abs/2303.10993}
}

@inproceedings{ChenLLLZS20,
  author       = {Deli Chen and
                  Yankai Lin and
                  Wei Li and
                  Peng Li and
                  Jie Zhou and
                  Xu Sun},
  title        = {Measuring and Relieving the Over-Smoothing Problem for Graph Neural Networks from the Topological View},
  booktitle    = {AAAI},
  year         = {2020}
}

@article{abs-2310-06825,
  author       = {Albert Q. Jiang and
                  Alexandre Sablayrolles and
                  Arthur Mensch and
                  Chris Bamford and
                  Devendra Singh Chaplot and
                  Diego de Las Casas and
                  Florian Bressand and
                  Gianna Lengyel and
                  Guillaume Lample and
                  Lucile Saulnier and
                  L{\'{e}}lio Renard Lavaud and
                  Marie{-}Anne Lachaux and
                  Pierre Stock and
                  Teven Le Scao and
                  Thibaut Lavril and
                  Thomas Wang and
                  Timoth{\'{e}}e Lacroix and
                  William El Sayed},
  title        = {Mistral 7B},
  journal      = {CoRR},
  year         = {2023},
  volume       = {abs/2310.06825}
}

@inproceedings{KangM18,
  author       = {Wang{-}Cheng Kang and
                  Julian J. McAuley},
  title        = {Self-Attentive Sequential Recommendation},
  booktitle    = {ICDM},
  year         = {2018}
}

@inproceedings{PathakGM17,
  author       = {Apurva Pathak and
                  Kshitiz Gupta and
                  Julian J. McAuley},
  title        = {SIGIR},
  year         = {2017}
}

@article{abs-2403-03952,
  author       = {Yupeng Hou and
                  Jiacheng Li and
                  Zhankui He and
                  An Yan and
                  Xiusi Chen and
                  Julian J. McAuley},
  title        = {Bridging Language and Items for Retrieval and Recommendation},
  journal      = {CoRR},
  year         = {2024},
  volume       = {abs/2403.03952}
}

@inproceedings{0003CSWC24,
  author       = {An Zhang and
                  Yuxin Chen and
                  Leheng Sheng and
                  Xiang Wang and
                  Tat{-}Seng Chua},
  title        = {On Generative Agents in Recommendation},
  booktitle    = {SIGIR},
  year         = {2024}
}

@inproceedings{NiLM19,
  author       = {Jianmo Ni and
                  Jiacheng Li and
                  Julian J. McAuley},
  title        = {Justifying Recommendations using Distantly-Labeled Reviews and Fine-Grained
                  Aspects},
  booktitle    = {EMNLP},
  year         = {2019}
}

@inproceedings{YouCSCWS20,
  author       = {Yuning You and
                  Tianlong Chen and
                  Yongduo Sui and
                  Ting Chen and
                  Zhangyang Wang and
                  Yang Shen},
  title        = {Graph Contrastive Learning with Augmentations},
  booktitle    = {NeurIPS},
  year         = {2020}
}

@article{abs-1809-10341,
  author       = {Petar Velickovic and
                  William Fedus and
                  William L. Hamilton and
                  Pietro Li{\`{o}} and
                  Yoshua Bengio and
                  R. Devon Hjelm},
  title        = {Deep Graph Infomax},
  journal      = {CoRR},
  year         = {2018},
  volume       = {abs/1809.10341}
}

@article{abs-2006-04131,
  author       = {Yanqiao Zhu and
                  Yichen Xu and
                  Feng Yu and
                  Qiang Liu and
                  Shu Wu and
                  Liang Wang},
  title        = {Deep Graph Contrastive Representation Learning},
  journal      = {CoRR},
  year         = {2020},
  volume       = {abs/2006.04131}
}

@inproceedings{SarwarKKR01,
  author       = {Badrul Munir Sarwar and
                  George Karypis and
                  Joseph A. Konstan and
                  John Riedl},
  title        = {Item-based collaborative filtering recommendation algorithms},
  booktitle    = {WWW},
  year         = {2001}
}

@article{LindenSY03,
  author       = {Greg Linden and
                  Brent Smith and
                  Jeremy York},
  title        = {Amazon.com Recommendations: Item-to-Item Collaborative Filtering},
  journal      = {{IEEE} Internet Comput.},
  year         = {2003},
  volume       = {7},
  number       = {1},
  pages        = {76--80}
}

@inproceedings{Karypis01,
  author       = {George Karypis},
  title        = {Evaluation of Item-Based Top-N Recommendation Algorithms},
  booktitle    = {CIKM},
  year         = {2001}
}

@inproceedings{VeitKBMBB15,
  author       = {Andreas Veit and
                  Balazs Kovacs and
                  Sean Bell and
                  Julian J. McAuley and
                  Kavita Bala and
                  Serge J. Belongie},
  title        = {Learning Visual Clothing Style with Heterogeneous Dyadic Co-Occurrences},
  booktitle    = {ICCV},
  year         = {2015}
}

@inproceedings{LiuZZ023,
  author       = {Yaokun Liu and
                  Xiaowang Zhang and
                  Minghui Zou and
                  Zhiyong Feng},
  title        = {Co-occurrence Embedding Enhancement for Long-tail Problem in Multi-Interest
                  Recommendation},
  booktitle    = {RecSys},
  year         = {2023}
}

@inproceedings{ZhouZY23,
  author       = {Zhihui Zhou and
                  Lilin Zhang and
                  Ning Yang},
  title        = {Contrastive Collaborative Filtering for Cold-Start Item Recommendation},
  booktitle    = {WWW},
  year         = {2023}
}

@article{abs-2103-00020,
  author       = {Alec Radford and
                  Jong Wook Kim and
                  Chris Hallacy and
                  Aditya Ramesh and
                  Gabriel Goh and
                  Sandhini Agarwal and
                  Girish Sastry and
                  Amanda Askell and
                  Pamela Mishkin and
                  Jack Clark and
                  Gretchen Krueger and
                  Ilya Sutskever},
  title        = {Learning Transferable Visual Models From Natural Language Supervision},
  journal      = {CoRR},
  year         = {2021},
  volume       = {abs/2103.00020}
}

@inproceedings{KarpukhinOMLWEC20,
  author       = {Vladimir Karpukhin and
                  Barlas Oguz and
                  Sewon Min and
                  Patrick Lewis and
                  Ledell Wu and
                  Sergey Edunov and
                  Danqi Chen and
                  Wen{-}tau Yih},
  title        = {Dense Passage Retrieval for Open-Domain Question Answering},
  booktitle    = {EMNLP},
  year         = {2020}
}

@inproceedings{0046ZL21,
  author       = {Liang Wang and
                  Wei Zhao and
                  Jingming Liu},
  title        = {Aligning Cross-lingual Sentence Representations with Dual Momentum
                  Contrast},
  booktitle    = {EMNLP},
  year         = {2021}
}

@inproceedings{ChenTZ0SZWC24,
  author       = {Yuxin Chen and
                  Junfei Tan and
                  An Zhang and
                  Zhengyi Yang and
                  Leheng Sheng and
                  Enzhi Zhang and
                  Xiang Wang and
                  Tat{-}Seng Chua},
  title        = {On Softmax Direct Preference Optimization for Recommendation},
  booktitle    = {NeurIPS},
  year         = {2024}
}

@article{abs-2312-11805,
  author       = {Gemini Team},
  title        = {Gemini: {A} Family of Highly Capable Multimodal Models},
  journal      = {CoRR},
  year         = {2023},
  volume       = {abs/2312.11805}
}

@article{abs-2506-05176,
  author       = {Yanzhao Zhang and
                  Mingxin Li and
                  Dingkun Long and
                  Xin Zhang and
                  Huan Lin and
                  Baosong Yang and
                  Pengjun Xie and
                  An Yang and
                  Dayiheng Liu and
                  Junyang Lin and
                  Fei Huang and
                  Jingren Zhou},
  title        = {Qwen3 Embedding: Advancing Text Embedding and Reranking Through Foundation
                  Models},
  journal      = {CoRR},
  year         = {2025},
  volume       = {abs/2506.05176}
}

@article{MalkovY20,
  author       = {Yury A. Malkov and
                  Dmitry A. Yashunin},
  title        = {Efficient and Robust Approximate Nearest Neighbor Search Using Hierarchical
                  Navigable Small World Graphs},
  journal      = {{IEEE} Trans. Pattern Anal. Mach. Intell.},
  year         = {2020},
  volume       = {42},
  number       = {4},
  pages        = {824--836}
}

@inproceedings{ParkJK17,
  author       = {Haekyu Park and
                  Jinhong Jung and
                  U Kang},
  title        = {A comparative study of matrix factorization and random walk with restart
                  in recommender systems},
  booktitle    = {BigData},
  year         = {2017}
}

@inproceedings{LeeKSJ24,
  author       = {Seunghan Lee and
                  Geonwoo Ko and
                  Hyun{-}Je Song and
                  Jinhong Jung},
  title        = {MuLe: Multi-Grained Graph Learning for Multi-Behavior Recommendation},
  booktitle    = {CIKM},
  year         = {2024}
}

@inproceedings{LeeSKLL11,
  author       = {Sangkeun Lee and
                  Sang{-}il Song and
                  Minsuk Kahng and
                  Dongjoo Lee and
                  Sang{-}goo Lee},
  title        = {Random walk based entity ranking on graph for multidimensional recommendation},
  booktitle    = {RecSys},
  year         = {2011}
}

@inproceedings{JungPSK17,
  author       = {Jinhong Jung and
                  Namyong Park and
                  Lee Sael and
                  U Kang},
  title        = {BePI: Fast and Memory-Efficient Method for Billion-Scale Random Walk
                  with Restart},
  booktitle    = {SIGMOD},
  year         = {2017}
}

@inproceedings{JungJSK16,
  author       = {Jinhong Jung and
                  Woojeong Jin and
                  Lee Sael and
                  U Kang},
  title        = {Personalized Ranking in Signed Networks Using Signed Random Walk with
                  Restart},
  booktitle    = {ICDM},
  year         = {2016}
}

@inproceedings{KimKLJS25,
  author       = {Kyungho Kim and
                  Sunwoo Kim and
                  Geon Lee and
                  Jinhong Jung and
                  Kijung Shin},
  title        = {Multi-behavior Recommender Systems: {A} Survey},
  booktitle    = {PAKDD},
  year         = {2025}
}

@inproceedings{KimKLS25,
  author       = {Kyungho Kim and
                  Sunwoo Kim and
                  Geon Lee and
                  Kijung Shin},
  title        = {A Self-Supervised Mixture-of-Experts Framework for Multi-behavior
                  Recommendation},
  booktitle    = {CIKM},
  year         = {2025}
}

@inproceedings{LeeJ23,
  author       = {Jong{-}whi Lee and
                  Jinhong Jung},
  title        = {Time-Aware Random Walk Diffusion to Improve Dynamic Graph Learning},
  booktitle    = {AAAI},
  year         = {2023}
}

@article{KoJ24,
  author       = {Geonwoo Ko and
                  Jinhong Jung},
  title        = {Learning disentangled representations in signed directed graphs without
                  social assumptions},
  journal      = {Inf. Sci.},
  year         = {2024},
  volume       = {665},
  pages        = {120373}
}

@article{HeGKW17,
  author       = {Xiangnan He and
                  Ming Gao and
                  Min{-}Yen Kan and
                  Dingxian Wang},
  title        = {BiRank: Towards Ranking on Bipartite Graphs},
  journal      = {{IEEE} Trans. Knowl. Data Eng.},
  year         = {2017},
  volume       = {29},
  number       = {1},
  pages        = {57--71}
}

@article{ko2025personalized,
  title={Personalized Ranking on Cascading Behavior Graphs for Accurate Multi-Behavior Recommendation},
  author={Ko, Geonwoo and Jeon, Minseo and Jung, Jinhong},
  journal={arXiv preprint arXiv:2502.11335},
  year={2025},
  volume       = {abs/2502.11335}
}

@article{abs-2012-14191,
  author       = {Jinhong Jung and
                  Jaemin Yoo and
                  U Kang},
  title        = {Signed Graph Diffusion Network},
  journal      = {CoRR},
  year         = {2020},
  volume       = {abs/2012.14191}
}

@article{GuJSJ25,
  author       = {Gyeongmin Gu and
                  Minseo Jeon and
                  Hyun{-}Je Song and
                  Jinhong Jung},
  title        = {Effective and lightweight representation learning for signed bipartite
                  graphs},
  journal      = {Neural Networks},
  year         = {2025},
  volume       = {192},
  pages        = {107708}
}

@article{ChunLSJ24,
  author       = {Jaewan Chun and
                  Geon Lee and
                  Kijung Shin and
                  Jinhong Jung},
  title        = {Random walk with restart on hypergraphs: fast computation and an application
                  to anomaly detection},
  journal      = {Data Min. Knowl. Discov.},
  volume       = {38},
  number       = {3},
  pages        = {1222--1257},
  year         = {2024}
}

@inproceedings{LeeK25,
  author       = {Jaeri Lee and
                  U Kang},
  title        = {Context-aware Sequential Bundle Recommendation via User-specific Representations},
  booktitle    = {CIKM},
  year         = {2025},
}

@inproceedings{LeeKS24,
  author       = {Geon Lee and
                  Kyungho Kim and
                  Kijung Shin},
  title        = {Revisiting LightGCN: Unexpected Inflexibility, Inconsistency, and
                  {A} Remedy Towards Improved Recommendation},
  booktitle    = {RecSys},
  year         = {2024}
}

@inproceedings{000100L25,
  author       = {Minjin Choi and
                  Sunkyung Lee and
                  Seongmin Park and
                  Jongwuk Lee},
  title        = {Linear Item-Item Models with Neural Knowledge for Session-based Recommendation},
  booktitle    = {SIGIR},
  year         = {2025},
}

@inproceedings{KimP23,
  author       = {Sang{-}Hong Kim and
                  Ha{-}Myung Park},
  title        = {Efficient Distributed Approximate k-Nearest Neighbor Graph Construction
                  by Multiway Random Division Forest},
  booktitle    = {SIGKDD},
  year         = {2023},
}

@article{JohnsonDJ21,
  author       = {Jeff Johnson and
                  Matthijs Douze and
                  Herv{\'{e}} J{\'{e}}gou},
  title        = {Billion-Scale Similarity Search with GPUs},
  journal      = {{IEEE} Trans. Big Data},
  volume       = {7},
  number       = {3},
  pages        = {535--547},
  year         = {2021}
}

@article{WangZZWMLWLXTHMHW25,
  author       = {Fali Wang and
                  Zhiwei Zhang and
                  Xianren Zhang and
                  Zongyu Wu and
                  Tzuhao Mo and
                  Qiuhao Lu and
                  Wanjing Wang and
                  Rui Li and
                  Junjie Xu and
                  Xianfeng Tang and
                  Qi He and
                  Yao Ma and
                  Ming Huang and
                  Suhang Wang},
  title        = {A Comprehensive Survey of Small Language Models in the Era of Large
                  Language Models: Techniques, Enhancements, Applications, Collaboration
                  with LLMs, and Trustworthiness},
  journal      = {{ACM} Trans. Intell. Syst. Technol.},
  volume       = {16},
  number       = {6},
  pages        = {145:1--145:87},
  year         = {2025}
}

@inproceedings{Kjju25,
  author       = {Ka Hyun Park and Junghun Kim and Jinhong Jung and U Kang
},
  title        = {PiGLeT: Probabilistic Message Passing for Semi-supervised Link Sign Prediction
},
  booktitle    = {ICDM},
  year         = {2025},
}

@inproceedings{ParkBKKKKKL25,
  author       = {Seungcheol Park and
                  Jeongin Bae and
                  Beomseok Kwon and
                  Minjun Kim and
                  Byeongwook Kim and
                  Se Jung Kwon and
                  U Kang and
                  Dongsoo Lee},
  title        = {Unifying Uniform and Binary-coding Quantization for Accurate Compression
                  of Large Language Models},
  booktitle    = {ACL},
  year         = {2025}
}

\appendix
\vspace{-3mm}
\section{Limitations and Future Directions}
\label{appendix:limitations}

\smallsection{Dependency on textual metadata and LMs}
\method exhibits an inherent dependency on textual item metadata and pre-trained LMs.
This dependency arises as an unavoidable trade-off of its ID-free design, which replaces raw item IDs with LRs to address the limitations associated with ID-based modeling.
As a result, in scenarios where no textual metadata is available, \method, along with other LR-based methods, cannot be directly applied.
Moreover, the quality of the resulting representations is inherently bounded by the expressive power of the chosen LM (see Appendix~\ref{appendix:Effect_of_LM}).
Consequently, the performance of \method may vary across domains depending on both the quality of textual metadata, as well as the specific LM employed.
However, this work, including~\cite{Sheng0ZCWC25}, uses only item titles whose LRs provide beneficial signals for recommendation. 
As such titles are widely available in most domains, \method is expected to be practically applicable.
Nevertheless, how to effectively leverage other textual metadata (e.g., category, price, description, etc.) beyond item titles, as well as the trade-offs induced by different LM choices, has not yet been systematically studied, making this an important open question for future research.

\smallsection{Preprocessing overhead and scalability}
Our method exhibits additional preprocessing overhead and scalability considerations.
In the data augmentation process (Section~\ref{sec:method:preprocessing}), \method computes pairwise similarities among items, which currently incurs $O(n^2)$ time for $n$ items.
Moreover, as shown in Table~\ref{tab:runtime}, generating LRs of items using pre-trained LMs accounts for about 98--99\% of the preprocessing time, leading to higher preprocessing costs than ID-dependent methods.
These overheads primarily affect the initial construction stage and may become a bottleneck for extremely large-scale datasets.
Nevertheless, note that this cost is incurred only once and does not impact inference efficiency, as only the original-view encoder is used during inference.
For better scalability, the quadratic complexity of similarity computation can be alleviated by adopting approximate $k$-NN search methods (e.g., HNSW~\cite{MalkovY20}, IVF~\cite{JohnsonDJ21}, and MRDF~\cite{KimP23}). 
In addition, employing lightweight LMs~\cite{WangZZWMLWLXTHMHW25,ParkBKKKKKL25} can further reduce preprocessing time, while systematically analyzing the associated trade-offs in representation quality remains an important direction for future work.

\smallsection{Future directions}
This work focuses on top-$K$ item recommendation with static user–item interaction histories. A natural extension is to apply the proposed design principles  to more expressive settings, including sequential recommendation~\cite{KangM18}, session-based recommendation~\cite{000100L25}, multi-behavior modeling~\cite{LeeKSJ24,KimKLJS25,KimKLS25}, bundle recommendation~\cite{LeeK25}, 
temporal dynamics~\cite{LeeJ23}, and signed or explicit feedback~\cite{abs-2012-14191,KoJ24,GuJSJ25,Kjju25}.
Moreover, while behavior similarity is computed based on co-occurrence in this study, future work may explore alternative graph-based similarities, such as random-walk-based~\cite{JungPSK17,ChunLSJ24} or bi-ranking methods~\cite{HeGKW17,ko2025personalized}, to further examine their effectiveness or trade-offs within the proposed framework.

\section{Detailed Experimental Settings}
\label{appendix:settings}

\smallsection{Machine} We used a workstation equipped with an AMD 5955WX, an RTX 4090 (24GB VRAM), and 128GB of main memory.

\smallsection{Implementation} We implemented our \method in Python 3.9 using PyTorch 1.13 with CUDA 11.7.
For the baselines, we used their official open-source implementations from the following sources:
\begin{itemize}[leftmargin=4mm]
    \item{\textbf{MF-BPR}: \url{https://github.com/RUCAIBox/RecBole/blob/master/recbole/model/general_recommender/bpr.py}}
    \item{\textbf{FISM}\footnote{We adapted FISM to employ BPR loss instead of BCE loss, based on the original code available at the corresponding repository.}: \url{https://github.com/RUCAIBox/RecBole/blob/master/recbole/model/general_recommender/fism.py}}
    \item{\textbf{LightGCN}: \url{https://github.com/gusye1234/LightGCN-PyTorch}}
    \item{\textbf{XSimGCL}: \url{https://github.com/Coder-Yu/SELFRec}}
    \item{\textbf{RLMRec}: \url{https://github.com/HKUDS/RLMRec}}
    \item{\textbf{AlphaRec}: \url{https://github.com/LehengTHU/AlphaRec}}
\end{itemize}

\smallsection{Details on datasets}
Following previous studies~\cite{NiLM19, abs-2403-03952, Sheng0ZCWC25}, we used datasets from Amazon Reviews~\cite{abs-2403-03952}\footnote{\url{https://amazon-reviews-2023.github.io/data_processing/5core.html}} across several categories, including \movie, \book, \video, \baby, \health, and \beauty. 
For a fair comparison with AlphaRec~\cite{Sheng0ZCWC25}, we used the preprocessed \movie and \book versions provided in that work\footnote{\url{https://github.com/LehengTHU/AlphaRec}}. 
For the other datasets, we used the 5-core versions from~\cite{abs-2403-03952}.
We also used the 5-core version of the \steam reviews~\cite{PathakGM17}\footnote{\url{https://cseweb.ucsd.edu/~jmcauley/datasets.html\#steam_data}}.
We used item titles as text. 

\smallsection{Language representations}
\label{appendix:lr_explain}
For \movie and \book, we employed the preprocessed LRs provided by~\cite{Sheng0ZCWC25}, based on text-embedding-3-large with dimension $\dlm=\text{3,072}$, which were used in both AlphaRec and AlphaFree.
For the other datasets, where no LRs are preprocessed, we obtained them using LLaMA-3.1-8B~\cite{abs-2407-21783} ($\dlm=\text{4,096}$).
RLMRec used GPT 3.5 Turbo and ada-002 as suggested in the original work ($\dlm=\text{1,536}$).

\smallsection{Hyperparameter tuning}
\label{appendix:hyperparam_scope}
We used Adam optimizer for training.
For baselines, we used the hyperparameters reported in their respective papers. 
For our method, we fixed embedding dimension $d$ to 64, learning rate $\eta$ to 0.0005, batch size $|\batch|$ to 4096, weight decay to $10^{-6}$, and the number $n_s$ of negative samples to 256. 
We trained for up to 500 epochs ($T_{\text{epoch}}=500$) with early stopping based on the validation metric (patience = 20).
We performed a grid search to select optimal values for the key hyperparameters: 
the number of similar items $K_c \in \{1, 3, 5, 10\}$, 
the alignment temperature $\tau_a \in \{0.01, 0.05, 0.15, 0.2\}$, the recommendation temperature $\tau_r \in \{0.15, 0.2\}$, and the alignment loss weight $\lambda_\text{align} \in \{0.01, 0.05, 0.15, 0.2\}$.

\renewcommand{\arraystretch}{0.9} 
\begin{table}[t]
\vspace{-3mm}
\small
\caption{Frequently-used symbols}
\label{tab:symbols}
\begin{tabular}{c|l}
\toprule
\textbf{Symbol} & \textbf{Description} \\
\midrule
$\items$   & set of items, where $n=|\items|$ \\
$\interactions$ \& $\interactions^{+}$ & original and augmented interaction sets, resp.  \\
$\traincol$   & set of all interaction sets, where $\interactions \in \traincol$ \\
$\batch$   & mini-batch, i.e., $\batch \subseteq \traincol$ \\
$K_c$ & number of similar items \\
$p$ & positive item randomly sampled from $I$\\
$N_I$ &  set of negatives drawn from $\items\setminus\interactions$ \\
$\mathcal{S}_{i}$ & set of similar items to item $i$ \\
$\lr{i}$ \& $\lr{i}^{+} \in \mathbb{R}^{\dlm}$ & original and augmented LRs of item $i$ \\
$\emb{*}$ \& $\emb{*}^+ \in \mathbb{R}^{\times d}$ & original and augmented embeddings of $*\in\{i, \interactions\}$\\
$\simf{B}(\cdot,\cdot)$ & behavioral similarity between two items \\
$\simf{S}(\cdot,\cdot)$ & semantic similarity between two items \\
$\dlm$ \& $d$  & dimension of LRs and embeddings \\
$\lambda_{\textnormal{align}}$  & ratio of the alignment loss \\
\bottomrule
\end{tabular}
\end{table}

\smallsection{Statistics on user cold-start setting}
As described in Section~\ref{sec:exp:generalization_unseen_interaction}, the split produces a user cold-start setting, where the average number of items per interaction set in the validation and test data ranges from 1.03 to 2.95, as shown in the table below.

\renewcommand{\arraystretch}{0.8} 
\begin{table}[h]
\setlength{\tabcolsep}{13pt}
\centering
\caption{Statistics on the number of items per interaction set in user cold-start setting.}
\begin{tabular}{c|ccc}
\hline
\toprule
\textbf{Datasets} & \textbf{Average} & \textbf{Maximum} & \textbf{Minimum} \\
\midrule
\movie & 2.95 & 142 & 1 \\
\book  & 2.70 & 74  & 1 \\
\video & 1.02 & 18  & 1 \\
\baby  & 1.03 & 19  & 1 \\
\bottomrule
\hline
\end{tabular}
\label{tab:item_stats}
\end{table}

\section{Detailed Algorithms}
\label{appendix:algorithms}
We provide detailed algorithms for each phase of \method in Algorithms~\ref{alg:preprocessing}–\ref{alg:inference}.
The preprocessing phase (Algorithm~\ref{alg:preprocessing}) is expressed in an item-wise manner rather than directly following the equations.
The similar-item search can be parallelized by performing a parallel loop over items (line~\ref{alg:preprocessing:i}), provided that the resulting item-level outputs are aggregated via an interaction-wise union operation (lines \ref{alg:preprocessing:i_aug}–\ref{alg:preprocessing:i_aug_union}).
Note that computing $\simf{B}(i,j)$ is equivalent to counting the number of users who have co-consumed both items $i$ and $j$.
Thus, it can be easily implemented in a vectorized manner as $\simf{B}(i,j) = \mat{B}_{ij}$, where $\mat{B} = \matt{A}\mat{A}$ and $\mat{A}$ denotes the user–item interaction matrix constructed from $\traincol$. 
In our implementation, we compute $\simf{B}$ directly in this matrix form, representing $\mat{A}$ as a sparse matrix (e.g., CSR).
In the training phase, although we describe the computation of embeddings only for the items in $I$, $I^+$, and $N_I$ in Section~\ref{sec:method:training:encoders}, our implementation precomputes embeddings for all items once before each batch to take advantage of parallel computation efficiency (line~\ref{alg:training:emb} of Alg.~\ref{alg:training}).

\renewcommand{\algorithmicrequire}{\textbf{Input:}}
\renewcommand{\algorithmicensure}{\textbf{Output:}}

\begin{algorithm}[t]
    \small
    \caption{Preprocessing phase of \method}
    \label{alg:preprocessing}
    \begin{algorithmic}[1]
        \Require 
            \item[-] $\traincol$: training collection of interaction sets
            \item[-] $\{t_i \mid i \in \items\}$: set of item contents
            \item[-] $K_c$: number of similar items
        \Ensure 
            \item[-]  $\{(I \mapsto I^+) \mid I \in \traincol\}$: collection of inter. sets with augmentations
            \item[-] $\{(\lr{i} \mapsto \lr{i}^{+}) \mid i \in \items\}$: set of item LRs with augmentations
        \vspace{1mm}    

        \State compute $\lr{i} \gets \texttt{LM}(t_i)$ \textbf{for each} $i \in \items$ \label{alg:preprocessing:lr}
        \State compute $H(i) \gets \{I \mid i \in I, I \in \traincol \}$ \textbf{for each} $i \in \items$ \label{alg:preprocessing:h}
        
        \State initialize $I^{+} \gets I$ \textbf{for each} $I \in \traincol$

        \For{\textbf{each} $i \in \items$} \label{alg:preprocessing:i}
            \For{\textbf{each} $j \in \items$} \label{alg:preprocessing:compute_sim}
                \State compute $\simf{B}(i, j) \gets |H(i) \cap H(j)|$ \Comment{{\scriptsize skip if $|H(i)|\cdot|H(j)|=0$}} \label{alg:preprocessing:sim_behavior}
                \State compute $\simf{S}(i, j) \gets \lr{i}^{\top}\lr{j}$ \label{alg:preprocessing:sim_semantic}
            \EndFor
            \State compute $\mu_i$ and $\cand{i}$ with $K_c$ according to Eqs.~\eqref{eq:cand:c}~and~\eqref{eq:sim:bs} \label{alg:preprocessing:cand}
            \For{\textbf{each} $I \in H(i)$} \label{alg:preprocessing:i_aug}
                \State $I^{+} \gets I^{+} \cup \cand{i}$ \label{alg:preprocessing:i_aug_union} 
            \EndFor
            \State compute $\lr{i}^+$ according to Eq.~\eqref{eq:representation_aug} \label{alg:preprocessing:lr_aug}
        \EndFor
        \State map $\lr{i} \mapsto \lr{i}^{+}$ \textbf{for each} $i \in \items$, and $I \mapsto I^{+}$  \textbf{for each} $I \in \traincol$ 

        \State \textbf{return} $\{(I \mapsto I^+) \mid I \in \traincol\}$ and $\{(\lr{i} \mapsto \lr{i}^{+}) \mid i \in \items\}$
    \end{algorithmic}
\end{algorithm}

\begin{algorithm}[t]
\caption{Training phase of \method}
\label{alg:training}
\small
\begin{algorithmic}[1]
    \Require
        \item[-] $\{(I \mapsto I^+) \mid I \in \traincol\}$:  collection of inter. sets with augmentations
        \item[-] $\{(\lr{i} \mapsto \lr{i}^{+}) \mid i \in \items\}$: set of item LRs with augmenations
        \item[-] $\tau_r$ and $\tau_a$: temperature hyperparameters
        \item[-] $\lambda_{\text{align}}$: ratio of the alignment loss 
        \item[-] $T_{\text{epoch}}$ and $\eta$: numbers of epochs and learning rate, resp.
    \Ensure 
        \item[-] $\theta_{\texttt{MLP}}:$ set of parameters of the orignal-view encoder
    \vspace{1mm}

    \State initialize model parameters in $\Theta = \{\theta_{\texttt{MLP}}, \theta_{\texttt{MLP}^{+}}\}$ \label{alg:training:init_param}

    \For{\textbf{each} epoch in $T_{\text{epoch}}$} \Comment{{\scriptsize $I^+ \mapsfrom I$ and $\lr{i}^+ \mapsfrom \lr{i}$}}
        \For{\textbf{each} mini-batch $\batch \subseteq \traincol$}
            \State  compute $\emb{i} \gets \texttt{MLP}(\lr{i})$ and $\emb{i}^{+} \gets \texttt{MLP}^{+}(\lr{i}^{+})$ \textbf{for each} $i \in \items$            \label{alg:training:emb} 
            \For{\textbf{each} $I \in \batch$} 
                \State compute $\emb{I} \gets \texttt{MLP}(
                    \texttt{mean}(\{\lr{i} \mid i \in I\}))$  \label{alg:training:emb_I}
                \State compute $\emb{I}^+ \gets \texttt{MLP}^{+}(
                    \texttt{mean}(\{\lr{i}^{+} \mid i \in I^+\}))$ \label{alg:training:emb_I+}
                \State randomly sample $p$ from $I$, and $N_I$ from $\items \setminus I$ \label{alg:training:sampling}
                \State compute local losses used for $\mathcal{L}_{\text{rec}}$ \Comment{{\scriptsize Eqs.~\eqref{eq:loss:rec:original}-\eqref{eq:loss:rec:augmented}}} \label{alg:training:loss_rec}
            \EndFor
            \State compute local losses used for $\mathcal{L}_{\text{align}}$ \Comment{{\scriptsize Eqs.~\scalebox{0.77}{\eqref{eq:loss:align:inter}}-\scalebox{0.77}{\eqref{eq:loss:align:item}}}} \label{alg:training:loss_align}
            \State compute $\mathcal{L}_{\text{final}}(\batch) \gets \mathcal{L}_{\text{rec}}(\batch) + \lambda_{\text{align}}\cdot\mathcal{L}_{\text{align}}(\batch)$ \Comment{{\scriptsize Eq.~\eqref{eq:loss:final}}}
            \State update $\theta \gets \theta - \eta\nabla_{\theta}\mathcal{L}_\text{final}(\batch)$ \textbf{for each} $\theta \in \Theta$
        \EndFor
    \EndFor
    \State \Return $\theta_{\texttt{MLP}}$
\end{algorithmic}
\end{algorithm}

\begin{algorithm}[t]
    \small
    \caption{Inference phase of \method}
    \label{alg:inference}
    \begin{algorithmic}[1]
        \Require 
            \item[-] $I_q$: querying interaction set
            \item[-] $\{\lr{i} \mid i \in \items\}$: set of item LRs
            \item[-] $\theta_{\texttt{MLP}}$: set of parameters of the original-view encoder
            \item[-] $K$: parameter for top-$K$ recommendation
        \Ensure 
            \item[-] $\mathcal{R}_{\interactions_q}$: top-$K$ ranked list recommended w.r.t. $I_q$
        \vspace{1mm}    

        \State compute $\emb{I_q} \gets \texttt{MLP}(\texttt{mean}(\{\lr{i} \mid i \in I_q\}))$ \label{alg:inference:mlpq}
        \State compute $\emb{i} \gets \texttt{MLP}(\lr{i})$ \textbf{for each} $i \in \items$ \label{alg:inference:mlpi}
        \State $\mathcal{R}_{\interactions_q} \gets \texttt{argtop}_{K}(
            \{ 
                \sigma\bigl(\emb{I_q}, \emb{i}\bigr) \mid i \in \items  
            \}
            )$ \label{alg:inference:topk}
        \Comment{{\scriptsize $y_{\theta_{\texttt{MLP}}}(i \mid \interactions_q) \coloneq \sigma\bigl(\emb{I_q}, \emb{i}\bigr)$}}
        \State \textbf{return} $\mathcal{R}_{\interactions_q}$
    \end{algorithmic}
\end{algorithm}

\section{Details of Semantic Filtering}
\label{appendix:semantic}

We analyzed the effect of semantic thresholds $\kappa$ in Section~\ref{sec:method:preprocessing:similar}, where Eq.~\eqref{eq:sim:bs} can be rewritten as $\cand{i} = \{
    j \in \mathcal{C}_{i} \mid \simf{S}(i, j) \geq \kappa
\}$.
We set the threshold $\kappa$ to correspond to the 25th, 50th, and 75th percentiles of the semantic similarity distribution, and measure retention rate (i.e., $(\sum_{i \in \items}|\cand{i}|)/(\sum_{i \in \items}|\mathcal{C}_{i}|)$), and Recall@20 of \method.
As shown in Table~\ref{tab:filtering-retention-basic}, the best performance is achieved when $\kappa$ corresponds to $\mu_i$, achieving a balanced selection of items based on behavioral similarity.
Note that in both datasets, we used different language models (i.e., Large 3 in \movie and LLAMA-3.1-8B in \baby), suggesting that our filtering strategy is potentially effective across different LMs.

\begin{table}[h]
\centering
\setlength{\tabcolsep}{11.5pt}
\caption{Effects of different semantic thresholds $\kappa$. }
\label{tab:filtering-retention-basic}
\begin{tabular}{cccc}
\hline
\toprule
\textbf{Datasets}  & \textbf{Threshold} $\kappa$ & \textbf{Retention} & \textbf{Recall@20} \\
\midrule
 & 25\% & 67.1\% & 0.1265 \\
\movie & \textbf{mean ($\mu_i$)} & \textbf{80.2\%} & \textbf{0.1267} \\
& 75\% & 91.6\% & 0.1264 \\
\midrule
 & 25\% & 43.9\% & 0.0414 \\
\baby & \textbf{mean ($\mu_i$)} & \textbf{63.9\%} & \textbf{0.0418} \\
& 75\% & 84.5\% & 0.0413 \\
\bottomrule
\hline
\end{tabular}
\end{table}


\section{Effect of Language Models}
\label{appendix:Effect_of_LM}
We examined the recommendation performance of \method across different LMs.
Depending on the number of LM parameters, lightweight LMs (e.g., Qwen-0.6B and Gemma-1B) produce embeddings with smaller $\dlm$, while medium-sized LMs (e.g., LLaMA-8B and Mist-7B) use larger $\dlm$, as shown in Table~\ref{tab:different_lms}. 
We compared \method and AlphaRec across different LMs, and further compared them with GNN-based methods to examine the role of explicit structural modeling.
From Table~\ref{tab:different_lms}, we observe the following:
\begin{itemize}[leftmargin=4mm]
    \item{
        With sufficiently expressive LMs such as LLaMA-8B and Mist-7B, rich LRs allow \method to effectively capture higher-order CF signals without explicit GNNs, enabling \method to achieve strong performance. 
     }
     \item {
        Under smaller LMs such as Qwen-0.6B and Gemma-1B, where expressiveness is rather limited, incorporating explicit structural biases via GNNs (e.g., AlphaRec or XSimGCL) remains beneficial.
     }
     \item {
        However, such GNN-based designs introduce substantial memory overhead, which becomes increasingly problematic as $\dlm$ grows, as shown in  Figure~\ref{fig:exper:vram_usage_LR_compare}.
     }
     \item {
        By eliminating GNNs, \method achieves superior memory efficiency while fully leveraging highly expressive LRs, resulting in state-of-the-art recommendation performance.
     }
\end{itemize}

\section{Validated Hyperparameters} 
We report the validated hyperparameters of \method for each dataset in Table~\ref{tab:hyperparams}.
\begin{table}[h]
\small
\setlength{\tabcolsep}{3pt}
\centering
\caption{Validated hyperparameters of \method. }
\begin{tabular}{c| c c c c c c c}
\hline
\toprule
\bf{Hyperparam.} & \movie & \book & \video & \baby & \steam & \beauty & \health \\
\midrule
$K_c$&5&5&10&10&3&10&5\\
$\lambda_\text{align}$&0.2&0.2&0.05&0.01&0.01&0.01&0.01\\
$\tau_a$&0.2&0.1&0.01&0.2&0.2&0.1&0.2\\
$\tau_r $&0.15&0.15&0.2&0.2&0.2&0.2&0.15\\
$\dlm$ & 3,072 & 3,072 & 4,096 & 4,096 & 4,096 & 4,096 & 4,096 \\
\bottomrule
\hline
\end{tabular}
\label{tab:hyperparams}
\end{table}

\begin{table}[t]
\scriptsize
\setlength{\tabcolsep}{1.5pt}
\renewcommand{\arraystretch}{1}
\centering

\definecolor{lightgray}{gray}{0.94}
\caption{Recommendation performance of \method across different LMs in terms of R@20 (Recall@20) and N@20 (NDCG@20), respectively.
For each LM, bold and underlined values denote the best and second-best scores, respectively.
The best R@20 and N@20 across all tested LMs are in red and blue, respectively.
Note that LightGCN and XSimGCL do not rely on LMs; therefore, their performance is identical across different LMs and is reported for comparison.
\label{tab:different_lms}
}
\label{tab:lm_all_vertical}
\resizebox{\linewidth}{!}{
\begin{tabular}{
c|c
>{\columncolor{lightgray}}c
>{\columncolor{lightgray}}c
cc
>{\columncolor{lightgray}}c
>{\columncolor{lightgray}}c
cc}
\hline
\toprule
\multirow{2}{*}[-0.45ex]{} & \multirow{2}{*}[-0.55ex]{\shortstack{\textbf{Methods}\\{\color{white}!}}}
& \multicolumn{2}{>{\columncolor{lightgray}}c}{\shortstack{\textbf{Qwen-0.6B}~\cite{abs-2506-05176} \\ ($\dlm$=1,024)}} 
& \multicolumn{2}{c}{\shortstack{\textbf{Gemma-1B}~\cite{abs-2312-11805} \\ ($\dlm$=1,152)}}
& \multicolumn{2}{>{\columncolor{lightgray}}c}{\shortstack{\textbf{LLaMA-8B}~\cite{abs-2407-21783} \\ ($\dlm$=4,096)}} 
& \multicolumn{2}{c}{\shortstack{\textbf{Mist-7B}~\cite{abs-2310-06825} \\ ($\dlm$= 4,096)}} \\
& & \textbf{R@20} & \textbf{N@20} & \textbf{R@20} & \textbf{N@20} & \textbf{R@20} & \textbf{N@20} & \textbf{R@20} & \textbf{N@20} \\
\midrule

\multirow{4}{*}{\rotatebox[origin=c]{90}{\video}}
& \textbf{\method} & \underline{0.103} & \underline{0.057} & \underline{0.094} & \underline{0.052} & \textbf{0.112} & \textbf{0.062} & \cellcolor{red!10}\textbf{0.118} & \cellcolor{blue!10}\textbf{0.065} \\
& \textbf{AlphaRec}  & \textbf{0.107} & \textbf{0.059} & \textbf{0.099} & \textbf{0.056} & \underline{0.109} & \underline{0.061} & \underline{0.115} & \underline{0.064} \\
& \textbf{LightGCN}  & 0.073 & 0.041 & 0.073 & 0.041 & 0.073 & 0.041 & 0.073 & 0.041 \\
& \textbf{XSimGCL}   & 0.090 & 0.050 & 0.090 & 0.050 & 0.090 & 0.050 & 0.090 & 0.050 \\
\midrule

\multirow{4}{*}{\rotatebox[origin=c]{90}{\baby}}
& \textbf{\method} & 0.034 & 0.018 & 0.034 & 0.018 & \textbf{0.041} & \textbf{0.022} & \cellcolor{red!10}\textbf{0.045} & \cellcolor{blue!10}\textbf{0.024} \\
& \textbf{AlphaRec}  & \textbf{0.040} & \textbf{0.021} & 0.034 & 0.018 & \underline{0.039} & \underline{0.021} & \underline{0.043} & \underline{0.023} \\
& \textbf{LightGCN}  & 0.036 & 0.019 & \underline{0.036} & \underline{0.019} & 0.036 & 0.019 & 0.036 & 0.019 \\
& \textbf{XSimGCL}   & \underline{0.039} & \underline{0.020} & \textbf{0.039} & \textbf{0.020} & 0.039 & 0.020 & 0.039 & 0.020 \\
\midrule

\multirow{4}{*}{\rotatebox[origin=c]{90}{\steam}}
& \textbf{\method} & \textbf{0.233} & \textbf{0.187} & \textbf{0.234} & \textbf{0.190} & \textbf{0.240} & \cellcolor{blue!10}\textbf{0.194} & \cellcolor{red!10}\textbf{0.241} & \textbf{0.193} \\
& \textbf{AlphaRec}  & \underline{0.226} & \underline{0.182} & \underline{0.230} & \underline{0.186} & \underline{0.236} & \underline{0.188} & \underline{0.238} & \underline{0.189} \\
& \textbf{LightGCN}  & 0.201 & 0.152 & 0.201 & 0.152 & 0.201 & 0.152 & 0.201 & 0.152 \\
& \textbf{XSimGCL}   & 0.225 & 0.166 & 0.225 & 0.166 & 0.225 & 0.166 & 0.225 & 0.166 \\
\bottomrule
\hline
\end{tabular}
}
\end{table}

\section{Computational Complexity Analysis}
\label{appendix:comp}

\smallsection{Time complexity analysis} We theoretically analyze the time complexity of each phase of \method, where Algorithms 1, 2, and 3 show the preprocessing, training, and inference phases, respectively.
Suppose 
$n = |\items|$ is the number of items, 
$m = \sum_{I \in \traincol}|I|$ is the total number of interactions, 
$h = |\traincol|$ is the number of interaction sets ($h \leq m$), 
$K_c$ is the number of similar items ($K_c \leq n$),
$|\batch|$ is the size of mini-batch, 
$\dlm$ is the dimension of LRs, and $d$ denotes the embedding dimensions of the MLP layers.

\begin{lemma}[Time Complexity for Preprocessing]
The time complexity of Algorithm~1 is $O(T_{\textnormal{LM}}n+ mn + (\dlm + \log{K_c})n^2)$, where $\tlm$ denotes the time at which the LM produces an embedding.
\end{lemma}
\begin{proof}
The algorithm takes $O(\tlm n)$ time as it computes the LRs of all items (line~1).
Constructing an inverted list for $H(i)$ takes $O(m)$ time, since $H(i)$ for each $i$ is built by iterating over all interactions in $\traincol$ (line~2).
Each intersection for $\simf{B}(\cdot)$ takes $O(|H(i)| + |H(j)|)$ time, assuming each $H(i)$ is sorted, and then it takes $\sum_{i \leq j}(|H(i)| + |H(j)|) = (n+1)\sum_{i}|H(i)|=O(nm)$ time (line~6).
Computing all $\simf{S}(\cdot)$ takes $O(\dlm n^2)$.
Then, it consumes $O(n^2\log{K_c}+nK_c)$ time to find similar items $\cand{i}$ for all items (line~8), if we use min-heap for sorting.
It takes $O(mK_c)$ and $O(nK_c\dlm)$ for obtaining $I^+$ and $\lr{i}^{+}$, respectively (lines~9~-~11).
Combining all the above terms proves the claim.
\end{proof}

\begin{lemma}[Time Complexity for Training]
In Algorithm~2, each epoch takes $O((T_i d n + K_c m + d h)\dlm + (n_s + |\batch|)hd)$ time, where $T_i = h/|\batch|$ is the number of iterations over all batches, and $n_s$ is the number of negatives.
\end{lemma}
\begin{proof}
    We first analyze the time for processing each mini-batch $\batch$.
    It takes $O(\dlm d n)$ time to perform the MLP operations for all items (line~4). 
    It takes $O(\dlm(|I| + d))$ time for each $I$ and $O(\dlm(|I^{+}| + d))$ for each $I^{+}$ (lines~6~and~7), where $|I^{+}| = O(K_c |I|)$, thus the overall cost becomes $O(\dlm(K_c |I| + d))$.
    Sampling $p$ and $N_I$ takes $O(n_s)$ time in expectation when using rejection sampling (line~8). 
    Then, it takes $O(|\batch| n_s d)$ time for $\mathcal{L}_{\textnormal{rec}}$, and $O(|\batch|^2d)$ time for $\mathcal{L}_{\textnormal{align}}$  (lines~9~and~10).
    Considering all batches is represented as follows:
    \begin{align*}
        \sum_{\batch \subseteq \traincol}\!\!\!\Bigl(&\dlm d n + \sum_{I \in \batch}\dlm(K_c |I| + d) + n_s + |\batch| n_s d + |\batch|^2d\Bigr) \\
        &= T_i\dlm d n + \dlm K_c m + \dlm d h + T_i n_s + h n_s d + |\batch|hd \\
        &= O((T_i d n + K_c m + d h)\dlm + (n_s + |\batch|)hd,
    \end{align*}
    where $n_s \leq n$.
\end{proof}

\begin{lemma}[Time Complexity for Inference]
The time complexity of Algorithm~3 is $O(\dlm d n + n\log{K})$, where $K$ is the parameter for top-$K$ recommendation.
\end{lemma}
\begin{proof}
    The algorithm takes $O(\dlm(|I_q| + d))$ time to perform the averaging and MLP operations (line~1).
    For all items, it takes $O(\dlm d n)$ for computing their embeddings (line~2).    
    Then, it takes $O(dn + n\log K)$ time to compute the scores of all items and select the top-$K$ items using a min-heap (line~3).
    Combining all the above terms proves the claim.
\end{proof}

\smallsection{Space complexity analysis}
We analyze the space complexity of each phase.
\begin{lemma}[Space Complexity for Preprocessing]
    The preprocessing phase of \method consumes $O(\dlm n + K_c(m+n))$ space.
\end{lemma}
\begin{proof}
    The LRs $\{\lr{i}\}$ and $\{\lr{i}^{+}\}$ for all items $i$ require $O(\dlm n)$ space (lines~1~and~11).
    The $\{H(i)\}$ for all items uses $O(m)$ space (line~2).
    The set $I^+$ uses $O(K_c|I|)$ space, and thus for all $I \in \traincol$, it takes $O(K_c m)$ space  (line~10).
    To find and store $\cand{i}$, it uses $O(K_c)$ space when a heap is used, and $O(K_c n)$ in total for all items (line~8).
    Note that for each $i$, it computes similarity scores of $i$ for every $j$ in an in-place manner, requiring only $O(n)$ space without storing results from previous $i$’s  (lines~5-7).
    Combining all the above terms proves the claim.
\end{proof}

\begin{lemma}[Space Complexity for Training]
    The training phase of \method uses $O(m + d(n + \dlm + n_s + |\batch|))$ space.
\end{lemma}
\begin{proof}
    As an input, $\traincol$ uses $O(m)$ space.
    Each MLP layer takes $O(\dlm d)$ space (line~1).
    The hidden embeddings $\{\emb{i}\}$ and $\{\emb{i}^+\}$ use $O(dn)$ space (line~4).
    For line 9, it uses $O(d n_s)$ space for $\mathcal{E}_{I}$ and $\mathcal{E}_{I}^{+}$ used to compute the local losses for $\mathcal{L}_{\textnormal{rec}}$ in an in-place manner.
    For line~10, it uses $O(|\batch|d)$ space for $\mathcal{E}_{\batch}$, $\mathcal{E}_{\batch}^{+}$, $\mathcal{E}_{P_{\batch}}$, and $\mathcal{E}_{P_{\batch}}^{+}$ used to compute the local losses for $\mathcal{L}_{\textnormal{align}}$ in an in-place manner.
    Combining all the above terms proves the claim.
\end{proof}

\begin{lemma}[Space Complexity for Inference]
    The inference phase of \method uses $O((\dlm + d)n + K)$ space.
\end{lemma}
\begin{proof}
    Computing all $\{\emb{i}\}$ as well as $\emb{I_q}$ takes $O((\dlm+d) n)$ space  (line~2), while selecting the top-$K$ items requires $O(K)$ space if a heap is used (line~3).
\end{proof}

\section{Case Study for Over-Smoothing}
\label{appendix:case:over}
Figure~\ref{fig:case:over} shows a case study comparing AlphaRec and \method to investigate the over-smoothing issue on the \baby dataset.
For the heavy user (User~A) with many interactions ($|I_A|=191$), AlphaRec assigns uniformly high scores with very low variance, suggesting that its item embeddings have become excessively similar  (i.e., over-smoothed) due to GNNs.
However, \method constructs user representations without GNNs, thereby preserving embedding diversity and avoiding over-smoothing for User~A.
For the light user (User~B) with few interactions ($|I_B|=17$),  AlphaRec yields higher yet broadly distributed scores than \method, indicating that even under sparse interactions, it produces more similar but unstable embeddings, whereas \method shows less similarity with a more compact and stable distribution, suggesting robustness against such effects.

\begin{figure}[h]
    \centering
    \hspace{6mm}
    \includegraphics[width=0.92\linewidth]{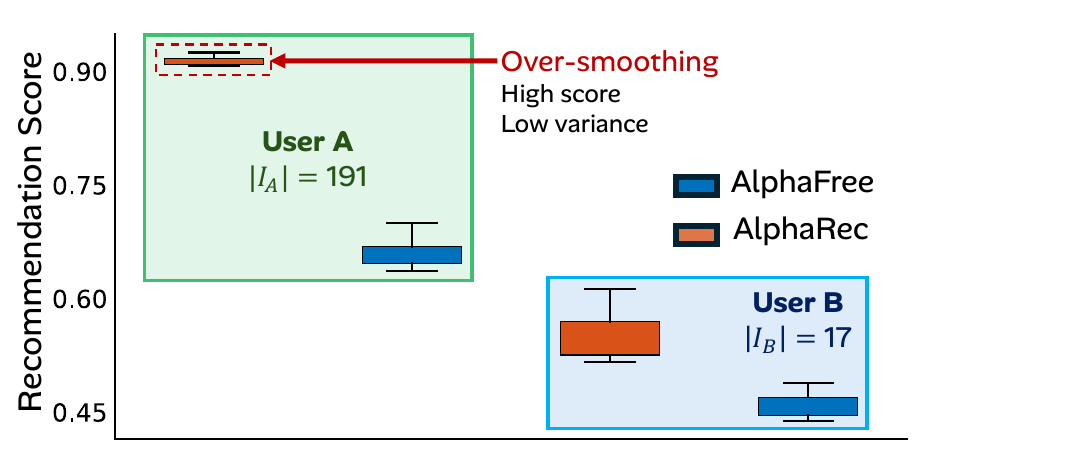}
    \caption{
    Case study of over-smoothing of AlphaRec and \method for heavy and light users on the \baby dataset.
    }
    \label{fig:case:over}
\end{figure}

\section{Experiments on Running Time}
\label{appendix:runtime}
In this section, we evaluated the running time of \method compared to LR-based methods (i.e., RLMRec and AlphaRec) and LightGCN across the preprocessing, training, and inference phases.
In the preprocessing phase, LR-based methods mainly compute item LRs using an LM, while other auxiliary operations vary across methods; in particular, \method computes similar items. 
LightGCN initializes embeddings and normalizes the user–item interaction matrix during preprocessing.
The preprocessing phase is executed only once, and we measured its wall-clock time $T_p$ in seconds.
For the training phase, we measured the average running time $T_t$ per epoch in seconds. 
During the inference phase, we measure the average time $T_i$ required to generate recommendations for a single user, reported in microseconds.

Table~\ref{tab:runtime} shows the results across different datasets. 
The preprocessing time of AlphaRec and that of \method are similar, as the LR generation part dominates this phase.
While their performances in training time are also comparable, AlphaRec suffers from out-of-memory errors on large datasets such as \beauty and \health (requiring over 24 GB of VRAM).
RLMRec runs out of time during preprocessing (exceeding 50 hours) on all evaluated datasets.
While LightGCN is the fastest in preprocessing and training, it exhibits higher inference latency due to GNN propagation and achieves lower recommendation performance than LR-based methods, as reported in Section~\ref{sec:exp:overall}.
From the inference-time perspective, \method achieves latency comparable to LightGCN. In contrast, AlphaRec relatively exhibits higher inference latency. 
It is worth noting that both AlphaRec and \method operate on the dimension $\dlm$ of LRs exceeding 3,000, whereas LightGCN performs computations using 64-dimensional embeddings.

\begin{table}[h!]
\scriptsize
\setlength{\tabcolsep}{0.4pt}
\renewcommand{\arraystretch}{0.97}
\definecolor{lightgray}{gray}{0.94}
\centering
\caption{
\label{tab:runtime}
Runtime comparison of \method with LR-based methods and LightGCN.
$T_{p}$ (sec), $T_{t}$ (sec), and $T_{i}$ ($\mu$s) indicate preprocessing, training (per epoch), and inference (per query user) times, respectively.
Parentheses under $T_{p}$ for LR-based methods indicate the ratio for generating LRs during preprocessing.
}
\begin{tabular}{
c
>{\columncolor{lightgray}}c
>{\columncolor{lightgray}}c
>{\columncolor{lightgray}}c
ccc
>{\columncolor{lightgray}}c
>{\columncolor{lightgray}}c
>{\columncolor{lightgray}}c
ccc
}

\hline
\toprule
\textbf{Methods}
& \multicolumn{3}{>{\columncolor{lightgray}}c}{\textbf{RLMRec}}
& \multicolumn{3}{c}{\textbf{AlphaRec}}
& \multicolumn{3}{>{\columncolor{lightgray}}c}{\textbf{\method}}
& \multicolumn{3}{c}{\textbf{LightGCN}}\\
\midrule
\textbf{Datasets}
& \makebox[2.0em][c]{$T_{p}$} & \makebox[1.5em][c]{$T_{t}$} & \makebox[1.8em][c]{$T_{i}$}
& \makebox[2.2em][c]{$T_{p}$} & \makebox[2.0em][c]{$T_{t}$} & \makebox[2.0em][c]{$T_{i}$}
& \makebox[2.2em][c]{$T_{p}$} & \makebox[1.8em][c]{$T_{t}$} & \makebox[2.0em][c]{$T_{i}$}
& \makebox[1.8em][c]{$T_{p}$} & \makebox[1.8em][c]{$T_{t}$} & \makebox[2.0em][c]{$T_{i}$} \\
\midrule
\video
& o.o.t. & \multicolumn{2}{@{}>{\columncolor{lightgray}}c@{}}{n/a}
& 11{\tiny(98\%)} & 5  & 19
& 563{\tiny(98\%)} & 9  & 14
& 10 & 1  & 13\\
\baby
& o.o.t. & \multicolumn{2}{@{}>{\columncolor{lightgray}}c@{}}{n/a}
& 797{\tiny(98\%)} & 12 & 23
& 802{\tiny(98\%)} & 13 & 18
& 39 & 1  & 18\\
\steam
& o.o.t. & \multicolumn{2}{@{}>{\columncolor{lightgray}}c@{}}{n/a}
& 342{\tiny(98\%)} & 91 & 28
& 315{\tiny(98\%)} & 45 & 3
& 84 & 8  & 9\\
\beauty
& o.o.t. & \multicolumn{2}{@{}>{\columncolor{lightgray}}c@{}}{n/a}
& 4614{\tiny(99\%)} & n/a & {\tiny(o.o.m.)}
& 4612{\tiny(99\%)} & 157 & 107
& 92 & 36  & 97\\
\health
& o.o.t. & \multicolumn{2}{@{}>{\columncolor{lightgray}}c@{}}{n/a}
& 4211{\tiny(99\%)} & n/a & {\tiny(o.o.m.)}
& 4180{\tiny(99\%)} & 161 & 93
& 563 & 39  & 90\\
\bottomrule
\hline
\end{tabular}
\end{table}


\end{document}